\documentclass[final,a4paper,UKenglish,thm-restate,autoref,numberwithinsect]{lipics-v2021}

\usepackage{xspace}
\usepackage{xcolor}
\usepackage{soul}
\usepackage[utf8]{inputenc}
\usepackage{amsmath,amssymb}
\usepackage{stmaryrd} 
\usepackage{wasysym}
\usepackage{etoolbox}

\newcommand{\resetCurThmBraces}{%
\gdef\curThmBraceOpen{(}%
\gdef\curThmBraceClose{)}}
\resetCurThmBraces
\newcommand{\removeThmBraces}{%
\gdef\curThmBraceOpen{}%
\gdef\curThmBraceClose{}}
\resetCurThmBraces

\newenvironment{notheorembrackets}{\removeThmBraces}{\resetCurThmBraces}

\usepackage{etoolbox}
\patchcmd{\thmhead}{(#3)}{\curThmBraceOpen #3\curThmBraceClose}{}{}

\usepackage{hyperref}
\hypersetup{hidelinks,final}

\usepackage[notref,notcite]{showkeys}
\usepackage[noadjust]{cite}

\usepackage{seqsplit}
\usepackage{xstring}
\newcommand{\defaultshowkeysformat}[1]{%
\StrSubstitute{#1}{ }{\textvisiblespace}[\TEMP]%
\parbox[t]{\marginparwidth}{\raggedright\normalfont\small\ttfamily\(\{\){\color{red!50!black}\expandafter\seqsplit\expandafter{\TEMP}}\(\}\)}%
}

\renewcommand*\showkeyslabelformat[1]{%
\noexpandarg%
\defaultshowkeysformat{#1}%
}

\newcommand{\midmid}{\hspace{0.2ex}{\rule[-0.1ex]{0.6pt}{1.65ex}}\hspace{0.2ex}}
\newcommand{\scriptmidmid}{\hspace{0.2ex}{\rule[-0.1ex]{0.6pt}{1.1ex}}\hspace{0.2ex}}

\newcommand{\compr}{:}

\newcommand{\newletter}[1]{{\midmid}#1}
\newcommand{\scriptnew}[1]{{\scriptmidmid}#1}

\newcommand{\F}{\mathcal{F}}
\newcommand{\FN}{\mathsf{FN}}

\renewcommand{\o}{\cdot}

\newcommand{\Names}[1]{\mathsf{N}(#1)}
\newcommand{\free}[1]{\mathsf{FN}(#1)}

\newcommand{\degree}{\mathsf{deg}}

\newcommand{\barA}{\overline{\names}}

\newcommand{\ub}{\mathsf{ub}}

\numberwithin{equation}{section}

\usepackage[author=anonymous,marginclue,footnote]{fixme}
\FXRegisterAuthor{dh}{adh}{DH}
\FXRegisterAuthor{ls}{als}{LS}
\FXRegisterAuthor{sm}{asm}{SM}
\FXRegisterAuthor{hu}{ahu}{HU}

\usepackage{algorithmicx}
    \usepackage[ruled]{algorithm}
\usepackage[noend]{algpseudocode}

\usepackage{tikz}
 \usetikzlibrary{cd}
 \usetikzlibrary{trees}
 \usetikzlibrary{shapes}
 \usetikzlibrary{fit}
 \usetikzlibrary{shadows}
 \usetikzlibrary{backgrounds}
 \usetikzlibrary{arrows,automata,positioning}

\tikzset{
   ->,
   n/.style= {circle,fill,inner sep=1.5pt,node distance=2cm}
  ,acc/.style={circle,draw,inner sep=3pt,node distance=2cm}
  ,phantom/.style={circle},
  ,arr/.style={->, >=stealth, semithick, shorten <= 3pt, shorten >= 3pt}
}

\newcommand{\CO}{\mathcal{O}}

\newcommand{\ol}{\overline}
\newcommand{\epito}{\twoheadrightarrow}
\newcommand{\monoto}{\hookrightarrow}
\newcommand{\xto}[1]{\xrightarrow{~#1~}}
\newcommand{\seq}{\subseteq}
\newcommand{\tl}{\widetilde}
\newcommand{\dash}{\mathop{-}}
\newcommand{\Perm}{\mathsf{Perm}}

\newcommand{\pow}{\mathcal{P}}

\makeatletter
\def\moverlay{\mathpalette\mov@rlay}
\def\mov@rlay#1#2{\leavevmode\vtop{%
   \baselineskip\z@skip \lineskiplimit-\maxdimen
   \ialign{\hfil$\m@th#1##$\hfil\cr#2\crcr}}}
\newcommand{\charfusion}[3][\mathord]{
    #1{\ifx#1\mathop\vphantom{#2}\fi
        \mathpalette\mov@rlay{#2\cr#3}
      }
    \ifx#1\mathop\expandafter\displaylimits\fi}
\makeatother

\newcommand{\dom}{\mathsf{dom}}

\newcommand{\braket}[1]{\langle #1 \rangle}

\renewcommand{\epsilon}{\varepsilon}

\newcommand{\PSPACE}{\mathrm{PSPACE}}
\newcommand{\orb}{\mathsf{orb}}
\newcommand{\fs}{\mathsf{fs}}

\newcommand{\Nat}{{\mathbb{N}}}

\newcommand{\Set}{\mathsf{Set}}

\newcommand{\op}{{\mathsf{op}}}
\newcommand{\Pow}{\mathcal{P}}
\newcommand{\ufs}{{\mathsf{ufs}}}

\DeclareMathOperator{\Nom}{\mathbf{Nom}}

\newcommand{\takeout}[1]{\empty}
\newcommand{\names}{\mathbb{A}}

\DeclareMathOperator{\supp}{\mathsf{supp}} 

\newcommand{\fresh}{\mathbin{\#}}

\newcommand{\powfs}{\pow_{\fs}}

\theoremstyle{definition}
\newtheorem{defn}[theorem]{Definition}
\newtheorem{expl}[theorem]{Example}
\newtheorem{rem}[theorem]{Remark}
\newtheorem{nota}[theorem]{Notation}

\title{Nominal Büchi Automata with Name Allocation}
\author{Henning Urbat}{Friedrich-Alexander-Universit\"at Erlangen-N\"urnberg, Germany}%
{henning.urbat@fau.de}%
{https://orcid.org/0000-0002-3265-7168}%
{Supported by Deutsche Forschungsgemeinschaft (DFG) under project SCHR~1118/15-1}
\author{Daniel Hausmann}{Gothenburg University, Sweden}%
{daniel.hausmann@fau.de}%
{https://orcid.org/0000-0002-0935-8602}%
{Supported by Deutsche Forschungsgemeinschaft (DFG) under project MI~717/7-1
and by the European Research Council (ERC) under project 772459}
\author{Stefan Milius}%
{Friedrich-Alexander-Universit\"at Erlangen-N\"urnberg, Germany}%
{mail@stefan-milius.eu}%
{https://orcid.org/0000-0002-2021-1644}%
{Supported by Deutsche Forschungsgemeinschaft (DFG) under project MI~717/7-1}
\author{Lutz Schr\"oder}{Friedrich-Alexander-Universit\"at Erlangen-N\"urnberg, Germany}%
{lutz.schroeder@fau.de}%
{https://orcid.org/0000-0002-3146-5906}%
{Supported by Deutsche Forschungsgemeinschaft (DFG) under project \mbox{SCHR~1118/15-1}}
\authorrunning{H.~Urbat, D.~Hausmann, S.~Milius, L.~Schr\"oder}
\Copyright{Henning Urbat, Daniel Hausmann, Stefan Milius, Lutz Schr\"oder}
\nolinenumbers 

\begin{document}
\maketitle
\begin{abstract}
Infinite words over infinite alphabets serve as models of the temporal development of the allocation and (re-)use of resources over linear time. We approach $\omega$-languages over infinite alphabets in the setting of nominal sets, and  study languages of \emph{infinite bar strings}, i.e.~infinite sequences of names that feature binding of fresh names; binding corresponds roughly to reading  letters from input words in automata models with registers. We introduce \emph{regular nominal nondeterministic Büchi automata} (\emph{Büchi RNNAs}), an automata model for languages of infinite bar strings, repurposing the previously introduced \emph{RNNAs} over finite bar strings. Our machines feature explicit binding (i.e.~resource-allocating) transitions and process their input via a Büchi-type acceptance condition. They emerge from the abstract perspective on name binding given by the theory of nominal sets. As our main result we prove that, in contrast to most other nondeterministic automata models over infinite alphabets, language inclusion of Büchi RNNAs is decidable and in fact elementary. This makes Büchi RNNAs a suitable tool for applications in model checking.
\end{abstract}
\newpage
\section{Introduction}\label{sec:intro}

Classical automata models and formal languages for finite words over
finite alphabets have been extended to infinity in both directions:
Infinite words model the long-term temporal development of systems,
while infinite alphabets model \emph{data}, such as
nonces~\cite{KurtzEA07}, object identities~\cite{GrigoreEA13}, or
abstract resources~\cite{CianciaSammartino14}. We approach \emph{data
  $\omega$-languages}, i.e.~languages of infinite words over infinite
alphabets, in the setting of nominal sets~\cite{Pitts13} where elements
of sets are thought of as carrying (finitely many) \emph{names} from a
fixed countably infinite reservoir. Following the paradigm of nominal
automata theory~\cite{BojanczykEA14}, we take the set of names as the
alphabet; we work with infinite words containing explicit name binding
in the spirit of previous work on nominal formalisms for finite words
such as nominal Kleene algebra~\cite{GabbayCiancia11} and regular
nondeterministic nominal automata (RNNAs)~\cite{SchroderEA17}. Name
binding may be viewed as the allocation of resources, or (yet) more
abstractly as an operation that reads fresh names from the input. We
refer to such infinite words as \emph{infinite bar strings}, in honour
of the vertical bar notation we employ for name binding. In the
present paper, we introduce a notion of nondeterministic nominal
automata over infinite bar strings, and show that it admits inclusion
checking in elementary complexity.

\enlargethispage{5pt} Specifically, we reinterpret the mentioned RNNAs
to accept infinite rather than finite bar strings by equipping them
with a B\"uchi acceptance condition; like over finite alphabets,
automata with more expressive acceptance conditions including Muller
acceptance can be translated into the basic (nondeterministic) B\"uchi
model. Our mentioned main result then states that language inclusion
of B\"uchi RNNAs is decidable in parametrized polynomial
space~\cite{StockhusenTantau13}, with a parameter that may be thought
as the number of registers. This is in sharp contrast to other
nondeterministic automata models for infinite words over infinite
alphabets, which sometimes have decidable emptiness problems but whose
inclusion problems are typically either undecidable or of
prohibitively high complexity even under heavy restrictions (inclusion
is not normally reducible to emptiness since nondeterministic models
typically do not determinize, and in fact tend to fail to be closed
under complement); details are in the related work section.

Infinite bar strings can be concretized to infinite strings of names
(i.e.~essentially to data words) by interpreting name binding as
reading either globally fresh letters (in which case B\"uchi RNNAs may
essentially be seen as a variant of session automata~\cite{BolligEA14}
for infinite words) or locally fresh letters. Both interpretations
arise from disciplines of $\alpha$-renaming as known from
$\lambda$-calculus~\cite{barendregt85}, with global freshness corresponding to a
discipline of \emph{clean} naming where bound names are never
shadowed, and local freshness corresponding to an unrestricted naming
discipline that does allow shadowing. The latter implies that local
freshness can only be enforced w.r.t.~names that are expected to be
seen again later in the word. It is precisely this fairly
reasonable-sounding restriction that buys the comparatively low
computational complexity of the model, which on the other hand allows
full nondeterminism (and, e.g., accepts the language `some letter
occurs infinitely often', which is not acceptable by deterministic
register-based models) and unboundedly many registers. B\"uchi RNNAs
thus provide a reasonably expressive automata model for infinite data
words whose main reasoning problems are decidable in elementary
complexity.

\subparagraph*{Related Work} B\"uchi RNNAs generally adhere to the
paradigm of register automata, which in their original incarnation
over finite words~\cite{KaminskiFrancez94} are equivalent to the
nominal automaton model of nondeterministic orbit-finite
automata~\cite{BojanczykEA14}. Ciancia and Sammartino
\cite{CianciaSammartino14} study \emph{deterministic} nominal Muller
automata accepting 
infinite strings of names, and show 
Boolean closure and decidability of inclusion. This model is
incomparable to ours; details are in \autoref{sec:relation}. For
alternating register automata over infinite data words, emptiness is
undecidable even when only one register is allowed (which over finite
data words does ensure decidability)~\cite{DemriLazic09}. For the
one-register safety fragment of the closely related logic Freeze LTL,
inclusion (i.e.~refinement) is decidable, but even the special case of
validity is not primitive recursive~\cite{Lazic06}, in particular not
elementary.

Many automata models and logics for data words deviate rather
substantially from the register paradigm, especially models in the
vicinity of data automata~\cite{BojanczykEA11}, whose emptiness
problem is decidable but at least as hard as Petri net reachability,
which by recent results is not elementary~\cite{CzerwinskiEA21}, in
fact Ackermann-complete~\cite{Leroux21,CzerwinskiOrlikowski21}. For
weak B\"uchi data automata~\cite{KaraEA12}, the emptiness problem is
decidable in elementary complexity. Similarly, Büchi generalized data
automata~\cite{ColcombetManuel14} have a decidable emptiness problem;
their B\"uchi component is deterministic. (Throughout, nothing appears
to be said about inclusion of data automata.) Variable finite
automata~\cite{GrumbergEA10} apply to both finite and infinite words;
in both versions, the inclusion problem of the nondeterministic
variant is undecidable.

\section{Preliminaries: Nominal Sets}
\label{sec:prelim}
Nominal sets form a convenient formalism for dealing with names and
freshness; for our present purposes, names play the role of data. We
briefly recall basic notions and facts and refer to~\cite{Pitts13} for a
comprehensive introduction. Fix a countably infinite set $\names$ of \emph{names}, and let $\Perm(\names)$
denote the group of finite permutations on $\names$, which is
generated by the \emph{transpositions} $(a\, b)$ for $a\neq b\in\names$
(recall that $(a\, b)$ just swaps~$a$ and~$b$). A \emph{nominal set} is a
set~$X$ equipped with a (left) group action $\Perm(\names)\times X\to X$, denoted $(\pi,x)\mapsto \pi\o x$, such that
every element $x\in X$ has a finite \emph{support}~$S\subseteq\names$,
i.e.~$\pi\cdot x=x$ for every $\pi\in \Perm(\names)$ such that $\pi(a)=a$ for all
$a\in S$. Every element~$x$ of a nominal set $X$ has a least finite
support, denoted $\supp(x)$. Intuitively, one should think of $X$ as a set of syntactic objects (e.g.~strings, $\lambda$-terms, programs), and of $\supp(x)$ as the set of names
 needed to describe an element $x\in X$. A name $a\in\names$ is \emph{fresh}
for~$x$, denoted $a\fresh x$, if $a\notin\supp(x)$. The \emph{orbit} of an element $x\in X$ is given by $\{ \pi\o x: \pi\in\Perm(\names)\}$. The orbits form a partition of $X$. The nominal set $X$ is \emph{orbit-finite} if it has only finitely many orbits.

Putting $\pi\cdot a = \pi(a)$ makes
$\names$ into a nominal set. Moreover,~$\Perm(\names)$ acts on subsets
$A\subseteq X$ of a nominal set~$X$ by
$\pi\cdot A = \{\pi\cdot x \compr x \in A\}$. A subset $A\seq X$ is \emph{equivariant} if $\pi\o A=A$ for all $\pi\in \Perm(\names)$. More generally, it is \emph{finitely
  supported} if it has finite support w.r.t.\ this action, i.e.~there exists a finite set $S\seq \names$ such that $\pi\o A = A$ for all $\pi\in \Perm(\names)$ such that $\pi(a)=a$ for all $a\in S$. The set $A$ is \emph{uniformly finitely supported} if $\bigcup_{x\in A} \supp(x)$ is a finite set. This implies that $A$ is finitely supported, with least support $\supp(A)=\bigcup_{x\in A} \supp(x)$~\cite[Theorem~2.29]{gabbay2011}. (The converse does not
hold, e.g.~the set $\names$ is finitely supported but not uniformly
finitely supported.) Uniformly finitely supported
orbit-finite sets are always finite (since an orbit-finite set contains only finitely many elements with a given finite support).  We respectively denote by $\Pow_\ufs(X)$ and $\powfs(X)$ the nominal sets of (uniformly) finitely supported subsets of a nominal set $X$. 

A map $f\colon X\to Y$ between nominal sets is \emph{equivariant} if $f(\pi\o x)=\pi\o f(x)$ for all $x\in X$ and $\pi\in \Perm(\names)$. Equivariance implies $\supp(f(x))\seq \supp(x)$ for all $x\in X$. The function $\supp$ itself is
equivariant, i.e.~$\supp(\pi\cdot x)=\pi\cdot \supp(x)$ for
$\pi\in \Perm(\names)$.  Hence $|\supp(x_1)| = |\supp(x_2)|$ whenever $x_1$, $x_2$
are in the same orbit of a nominal set. 

We denote by $\Nom$ the category of nominal sets and equivariant maps. The object maps $X\mapsto \Pow_\ufs(X)$ and $X\mapsto \powfs(X)$ extend to endofunctors $\Pow_\ufs\colon \Nom\to \Nom$ and $\powfs\colon \Nom\to \Nom$
sending an equivariant map $f\colon X\to Y$ to the map $A\mapsto f[A]$. 

The coproduct $X+Y$ of nominal sets $X$ and $Y$ is given by their disjoint union with the group action inherited from the two summands. Similarly, the product $X \times Y$ is given by the cartesian product with the componentwise group action; we have $\supp(x,y)
= \supp(x) \cup \supp(y)$. Given a nominal set $X$ equipped with an
equivariant equivalence relation, i.e.~an equivalence relation $\sim$ that is equivariant as a 
subset $\mathord{\sim} \subseteq X \times X$, the quotient
$X/\mathord{\sim}$ is a nominal set under the expected group action
defined by $\pi \cdot [x]_\sim = [\pi \cdot x]_\sim$.

A key role in the technical development is played by
\emph{abstraction sets}, which provide a semantics for binding
mechanisms~\cite{GabbayPitts99}. Given a nominal set $X$, an equivariant equivalence relation $\sim$ on
  $\names \times X$ is defined by $(a,x)\sim (b,y)$ iff
  $(a\, c)\cdot x=(b\, c)\cdot y$ for some (equivalently, all)
  fresh~$c$. The \emph{abstraction set} $[\names]X$ is the quotient
  set $(\names\times X)/\mathord{\sim}$. The $\sim$-equivalence class
  of $(a,x)\in\names\times X$ is denoted by
  $\braket{a} x\in [\names]X$. We may think of~$\sim$ as an abstract notion of $\alpha$-equivalence,
and of~$\braket{a}$ as binding the name~$a$. Indeed we have
$\supp(\braket{a} x)= \supp(x)\setminus\{a\}$ (while
$\supp(a,x)=\{a\}\cup\supp(x)$), as expected in binding constructs.

 The object map $X\mapsto [\names]X$ extends to an endofunctor
$[\names]\colon \Nom \to \Nom$
sending an equivariant map $f\colon X\to Y$ to the equivariant map $[\names]f\colon [\names]X\to [\names]Y$ given by $\braket{a}x\mapsto \braket{a}f(x)$ for $a\in \names$ and $x\in X$.


\section{The Notion of $\boldsymbol{\alpha}$-Equivalence for Bar Strings}\label{sec:alpha-equiv}
In the following we investigate automata consuming input words over the infinite alphabet 
\[
  \barA := \names \cup \{\newletter a: a\in \names\}.
\]
A \emph{finite bar string} is a finite word $\sigma_1\sigma_2\cdots \sigma_n$ over $\barA$, and \emph{infinite bar string} is an infinite word $\sigma_1\sigma_2 \sigma_3\cdots$ over $\barA$. We denote the sets of finite and infinite bar strings by  $\barA^*$ and $\barA^\omega$, respectively. Given $w\in \barA^*\cup \barA^\omega$ the set of \emph{names} in $w$ is defined by 
\[ \Names{w} = \{ a\in \names \compr \text{the letter $a$ or $\newletter a$ occurs in $w$} \}. \]
An infinite bar string $w$ is \emph{finitely supported} if $\Names{w}$ is a finite set; we let $\barA^\omega_\fs\seq \barA^\omega$ denote the set of finitely supported infinite bar strings. Note that $\barA^*$ and $\barA^\omega_\fs$ are nominal sets w.r.t.\ the group action defined pointwise. The least support of a bar string is its set of names.

A bar string containing only letters from $\names$ is called a \emph{data word}. We denote by $\names^*\seq \barA^*$, $\names^\omega\seq \barA^\omega$ and $\names^\omega_\fs\seq \barA^\omega_\fs$ the sets of finite data words, infinite data words, and finitely supported infinite data words, respectively. 

We interpret $\newletter a$ as binding the name $a$ to the right. Accordingly, a name $a\in \names$ is said to be \emph{free} in a bar string $w\in \barA^*\cup \barA^\omega$ if (i) the letter $a$ occurs in $w$, and (ii) the first occurrence of $a$ is not preceded by any occurrence of $\newletter a$. For instance, the name $a$ is free in $a\newletter aba$ but not free in $\newletter aaba$, while the name $b$ is free in both bar strings. We put 
\[\free{w} = \{ a\in \names\compr \text{$a$ is free in $w$} \}. \]
We obtain a natural notion of $\alpha$-equivalence for both finite and infinite bar strings; the finite case is taken from~\cite{SchroderEA17}.

\begin{defn}[$\alpha$-equivalence]\label{def:alpha-fin}
  Let $=_\alpha$ be the least equivalence relation on $\barA^*$ such that
  \[
    x\newletter av =_\alpha x\newletter bw
    \qquad
    \text{for all $a,b\in \names$ and $x,v,w\in \barA^*$ such that
      $\braket{a} v = \braket{b} w$}.
  \]
This extends to an equivalence relation $=_\alpha$ on $\barA^\omega$  given by \[v=_\alpha w \qquad\text{iff}\qquad v_n=_\alpha w_n \quad\text{for all $n\in \Nat$},\] 
where $v_n$ and $w_n$ are the prefixes of length $n$ of $v$ and $w$. We denote by $\barA^*/\mathord{=_\alpha}$ and $\barA^\omega/\mathord{=_\alpha}$ the sets of $\alpha$-equivalence classes of finite and infinite bar strings, respectively, and we write $[w]_\alpha$ for the $\alpha$-equivalence class of $w\in \barA^*\cup \barA^\omega$.
\end{defn}
\begin{rem}\label{rem:alpheq-technical}
  \begin{enumerate}
  \item For any $v,w\in \barA^*$ the condition $\braket{a} v = \braket{b} w$ holds if
    and only if
    \[
      \text{$a=b$ and $v=w$,}
      \qquad \text{or}\qquad
      \text{$b\fresh v$ and $(a\, b)\o v = w$.}
    \]
  \item The $\alpha$-equivalence relation is a \emph{left congruence}: 
    \[
      v=_\alpha w\quad\text{implies}\quad xv=_\alpha xw
      \qquad\qquad
      \text{for all $v,w\in \barA^*\cup \barA^\omega$ and $x\in \barA^*$.}
    \]
    Moreover, the \emph{right cancellation property} holds:
    \[
      vx=_\alpha wx \quad\text{implies}\quad v=_\alpha w
      \qquad\qquad
      \text{for all $v,w\in \barA^*$ and $x\in \barA^*\cup \barA^\omega$.}
    \]
  \item The equivalence relation $=_\alpha$ is equivariant. Therefore, both $\barA^*/{=_\alpha}$ and $\barA^\omega_\fs/{=_\alpha}$ are nominal sets with the group action $\pi\o [w]_\alpha = [\pi\o w]_\alpha$ for $\pi\in \Perm(\names)$ and $w\in \barA^*\cup \barA^\omega_\fs$.
  The least support of $[w]_\alpha$ is
    the set $\free{w}$ of free names of $w$.
  \end{enumerate}
\end{rem}

\begin{rem}
  Our notion of $\alpha$-equivalence on infinite bar strings differs
  from the equivalence relation generated by relating $x\newletter av$
  and $x\newletter bw$ whenever $\braket{a} v = \braket{b} w$ (as in
  \autoref{def:alpha-fin} but now for infinite bar strings
  $v,w\in \barA^\omega_\fs$): The latter equivalence relates two bar
  strings iff they can be transformed into each other by finitely many
  $\alpha$-renamings, while the definition of $\alpha$-equivalence as
  per \autoref{def:alpha-fin} allows infinitely many simultaneous
  $\alpha$-renamings; e.g.~the infinite bar string
  $(|aa)^\omega= \newletter aa \newletter aa\newletter aa\newletter aa
  \cdots$ is $\alpha$-equivalent to
  $(|aa|bb)^\omega = \newletter aa \newletter bb\newletter aa
  \newletter bb\cdots$.
\end{rem}

\begin{defn}
  A \emph{literal language}, \emph{bar language} or \emph{data
    language} is a subset of $\barA^*$, $\barA^*/{=_\alpha}$ or
  $\names^*$, respectively.  Similarly, a \emph{literal
    $\omega$-language}, \emph{bar $\omega$-language} or \emph{data
    $\omega$-language} is a subset of $\barA^\omega$,
  $\barA^\omega/{=_\alpha}$ or $\names^\omega$, respectively.
\end{defn}
\begin{nota}\label{not:local-semantics}
  Given a finite or infinite bar string
  $w\in \barA^*\cup \barA^\omega$ we let $\ub(w)$ denote the data word
  obtained by replacing every occurrence of $\newletter a$ in $w$ by $a$; e.g.\
  $\ub(\newletter aa\newletter a\newletter bb)=aaabb$. To every bar ($\omega$-)language $L$ we associate the data
  ($\omega$-)language $D(L)$ given by
  \[
    D(L) = \{ \ub(w) \compr [w]_\alpha \in L \}.
  \]
\end{nota}
\begin{defn}\label{D:clean}
  A finite or infinite bar string $w$ is \emph{clean} if for each
  $a\in \free{w}$ the letter~$\newletter a$ does not occur in $w$, and for each
  $a\not\in \free{w}$ the letter $\newletter a$ occurs at most once.
\end{defn}
\begin{lemma}\label{lem:clean}
  Every bar string $w\in \barA^*\cup \barA^\omega_\fs$ is $\alpha$-equivalent to a (not necessarily finitely supported) clean bar string.
\end{lemma}
\begin{proof}
  Since $w$ has finite support, for every occurrence of the letter
  $\newletter a$ for which $a$ or $\newletter a$ has already occurred before, one can
  replace $\newletter a$ by $\newletter b$ for some fresh name $b$ and replace the
  suffix~$v$ after $\newletter a$ by $(a\, b)v$. Iterating this 
  yields a clean bar string $\alpha$-equivalent to~$w$.
\end{proof}

\begin{expl}
\begin{enumerate}
\item The finitely supported infinite bar string $\newletter aa\newletter bb\newletter aa\newletter bb\cdots$
  is $\alpha$-equivalent to the clean bar string
  $\newletter a_1a_1\newletter a_2a_2\newletter a_3a_3\newletter a_4a_4\cdots$ where the $a_i$ are pairwise
  distinct names. Note that $\newletter aa\newletter bb\newletter aa\newletter bb\cdots$ is not $\alpha$-equivalent to any finitely supported clean bar string.
  
\item For non-finitely supported bar
  strings the lemma generally fails: if $\names = \{a_1,a_2,a_3,\cdots\}$ then the bar string
  $a_1\newletter a_1a_2a_3a_4a_5a_6a_7\cdots$ is not $\alpha$-equivalent to any other bar string, in particular not to a clean one.
\end{enumerate}
\end{expl}

\noindent For readers familiar with the theory of coalgebras we note that on infinite bar strings with finite support, $\alpha$-equivalence
naturally emerges from a coinductive point of view. Kurz et
al.~\cite{kpsv13} use coinduction to devise a general notion of
$\alpha$-equivalence for infinitary terms over a binding signature,
which form the final colgebra for an associated endofunctor on $\Nom$.
The following is the special case for the endofunctors
\[
  GX = \barA\times X \cong \names \times X + \names \times X
  \qquad\text{and}\qquad
  FX = \names\times X + [\names]X.
\]
Let $\nu F$ and $\nu G$ denote their respective final coalgebras.
The coalgebra $\nu G$ is carried by the nominal set
$\barA^\omega_\fs$ 
with the usual coalgebra structure
$\langle \mathsf{hd},\mathsf{tl}\rangle\colon \barA^\omega_\fs \to
\barA \times \barA^\omega_\fs$ decomposing an infinite string into its
head and tail. The natural transformation
\[
  \sigma_X\colon GX\twoheadrightarrow FX\qquad\text{given by}\qquad
  (a,x) \mapsto (a,x),\quad (\newletter {a},x) \mapsto \langle
  a\rangle x,
\]
then induces a canonical map $e_\alpha\colon \barA^\omega_\fs \to \nu F$,
viz.~the unique
homomorphism from the $F$-coalgebra
$\nu G \xto{\langle \mathsf{hd},\mathsf{tl}\rangle} G(\nu G)
\xto{\sigma_{\nu G}} F(\nu G)$ into the final coalgebra $\nu F$.
\begin{proposition}\label{lem:alpha-equiv-coalgebraic}
  For every $v,w\in \barA^\omega_\fs$ we have $v=_\alpha w$ if and
  only if $e_\alpha(v)=e_\alpha(w)$.
\end{proposition}
\section{Automata over Infinite Bar Strings}\label{sec:rnna}

We proceed to introduce {Büchi RNNAs}, our nominal automaton model for bar $\omega$-languages and data $\omega$-languages. It modifies regular nominal nondeterministic automata
(RNNAs)~\cite{SchroderEA17}, originally introduced for bar languages and data languages of finite words, to accept infinite words. Roughly, Büchi RNNAs are to RNNAs what classical Büchi automata~\cite{GraedelThomasWilke02} are to nondeterministic finite automata. We first recall:
\begin{defn}[RNNA~\cite{SchroderEA17}]
\begin{enumerate}
\item  An \emph{RNNA} $A=(Q,R,q_0,F)$ is given by an orbit-finite nominal set $Q$
  of \emph{states}, an equivariant relation
  $R\seq Q\times \barA \times Q$ specifying \emph{transitions}, an
  \emph{initial state}~ $q_0\in Q$ and an equivariant set $F\seq Q$ of
  \emph{final states}. We write $q\xto{\sigma}q'$ if
  $(q,\sigma,q')\in R$. The transitions are subject to two conditions:
  \begin{enumerate}[(1)]
  \item \emph{$\alpha$-invariance}: if $q\xto{\scriptnew a}q'$ and
    $\braket{a} q'=\braket{b} q''$, then $q\xto{\scriptnew b}q''$.
    
  \item \emph{Finite branching up to $\alpha$-invariance:} For each $q\in Q$ the sets $\{ (a,q') \compr q\xto{a}q'\}$ and
    $\{ \braket{a} q' \compr q\xto{\scriptnew a} q' \}$ are finite (equivalently, uniformly finitely supported).
  \end{enumerate}
  The \emph{degree} of $A$, denoted $\degree(A)$, is the maximum of
  all $|\supp(q)|$ where $q\in Q$.
\item  Given a finite bar string
    $w=\sigma_1\sigma_2\cdots \sigma_n\in \barA^*$ and a state $q\in Q$, a \emph{run} for
    $w$ from $q$ is a sequence of transitions 
    \[
      q\xto{\sigma_1}q_1\xto{\sigma_2}\cdots \xto{\sigma_n}q_n.
    \]
    The run is \emph{accepting} if $q_n$ is final. The state $q$ \emph{accepts} $w$ if there exists an accepting run for $w$ from $q$, and the automaton $A$ \emph{accepts} $w$ if its initial state $q_0$ accepts $w$.
    We define
    \begin{align*}
      L_0(A)
      & = \{w\in \barA^*\compr \text{$A$ accepts $w$}\},
      & \text{the \emph{literal language accepted by $A$},}
      \\
      L_{\alpha}(A)
      & = \{ [w]_\alpha \compr w\in \barA^*, \, \text{$A$ accepts $w$} \},
      & \text{the \emph{bar language accepted by $A$},}
      \\
      D(A)
      & = D(L_{\alpha}(A)),& \text{the \emph{data language accepted by $A$}}. 
    \end{align*}
\end{enumerate}
\end{defn}
\begin{rem}\label{rem:coalgebra}
  Equivalently, an RNNA is an orbit-finite coalgebra
  \[
    Q\longrightarrow 2\times \Pow_\ufs(\names\times Q)\times
    \Pow_\ufs([\names]Q)
  \]
  for the functor
  $FX = 2\times \Pow_\ufs(\names\times X)\times
  \Pow_\ufs([\names]X)$ on $\Nom$, 
  with an initial state $q_0\in Q$.
\end{rem}

\begin{defn}[Büchi RNNA]
 A \emph{Büchi RNNA} is an RNNA $A=(Q,R,q_0,F)$ used to recognize infinite bar strings as follows. Given
    $w=\sigma_1\sigma_2\sigma_3\cdots \in \barA^\omega$ and a state $q\in Q$, a \emph{run}
    for $w$ from $q$ is an infinite sequence of transitions
    \[
      q_0\xto{\sigma_1}q_1\xto{\sigma_2} q_2\xto{\sigma_3} \cdots.
    \]
    The run is \emph{accepting} if $q_n$ is final for infinitely many
    $n\in \Nat$. The state $q$ \emph{accepts} $w$ if there exists an accepting run for $w$ from $q$, and the automaton $A$ \emph{accepts} $w$ if its initial state $q_0$ accepts $w$. We define
    \begin{align*}
      L_{0,\omega}(A)
      & = \{w\in \barA^\omega\compr \text{$A$ accepts $w$}\}, &
      \text{the \emph{literal $\omega$-language accepted by $A$},}
      \\
      L_{\alpha,\omega}(A)
      & = \{ [w]_\alpha \compr w\in \barA^\omega, \, \text{$A$ accepts $w$}\},
      & \text{the \emph{bar $\omega$-language accepted by $A$},}
      \\
      D_\omega(A)
      & = D(L_{\alpha,\omega}(A)), 
      & \text{the \emph{data $\omega$-language accepted by $A$}}.
    \end{align*}
\end{defn}

\begin{example}\label{ex:rnna}
Consider the Büchi RNNA $A$ with states $\{q_0\} \cup \names\times \{0,1\}$ and transitions as displayed below, where $a,b$ range over distinct names in 
$\names$. Note that the second node represents the orbit $\names\times \{0\}$ and the third one the orbit $\names\times \{1\}$.
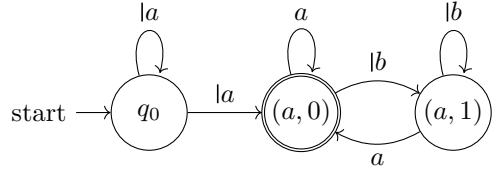
\begin{center}
\begin{tikzpicture}[align=center,node distance=2cm, state/.style={circle, draw, minimum size=1cm, inner sep=1pt}] 
\node[state, initial] (q0) {$q_0$};
\node[state, right of=q0, accepting] (a0) {$(a,0)$};
\node[state, right of=a0] (a1) {$(a,1)$};
\draw
(q0) edge[loop above] node{$\newletter a$} (q0)
(q0) edge[above] node{$\newletter a$} (a0)
(a0) edge[loop above] node{$a$} (a0)
(a0) edge[above, bend left] node{$\newletter b$} (a1)
(a1) edge[loop above] node{$\newletter b$} (a1)
(a1) edge[below, bend left] node{$a$} (a0);
\end{tikzpicture}
\end{center}
 The data $\omega$-language $D_\omega(A)$ consists of all infinite words $w\in \names^\omega$ where some name $a\in \names$ occurs infinitely often.
\end{example}

\begin{rem}
  For automata consuming infinite words over infinite alphabets, a slightly subtle point is whether to admit arbitrary infinite words
  as inputs or restrict to finitely supported ones. For the bar language semantics of Büchi RNNAs this choice is inconsequential: we shall see in
  \autoref{lem:restrict_num_names} below that modulo
  $\alpha$-equivalence all infinite bar strings accepted by Büchi RNNA
  are finitely supported. However, for the data language semantics admitting strings with infinite support is important for the decidability of language inclusion, see \autoref{rem:data-language-finite-support}.
\end{rem}
%
\begin{notheorembrackets}
  \begin{lemma}[{\cite[Lem.~5.4]{SchroderEA17}}]\label{lem:rnna-props}
    Let $A=(Q,R,q_0,F)$ be an RNNA, $q,q'\in Q$ and $a\in \names$.
    \begin{enumerate}
    \item\label{lem:rnna-props:1} If $q\xto{a} q'$ then $\supp(q')\cup \{a\}\seq \supp(q)$.
    \item\label{lem:rnna-props:2} If $q\xto{\scriptnew a} q'$ then $\supp(q')\seq \supp(q)\cup \{a\}$. 
    \item\label{lem:rnna-props:3} For every bar string $w\in \barA^*\cup \barA^\omega$ with a run from $q$ one has $\FN(w)\seq \supp(q)$.
    \end{enumerate}
  \end{lemma}
\end{notheorembrackets}

\noindent The next proposition will turn out to be crucial in the development that follows. It asserts that, up to $\alpha$-equivalence, one can always restrict the inputs of a Büchi RNNA to bar strings with a finite number of names, bounded by the degree of the automaton.
\begin{proposition}\label{lem:restrict_num_names}
  Let $A$ be a Büchi RNNA accepting the infinite bar string $w\in
  \barA^\omega$. Then it also accepts some $w'\in \barA^\omega_\fs$ such
  that
  \[
    w'=_\alpha w
    \qquad\text{and}\qquad
    |\supp (q_0) \cup \Names{w'}|\leq \degree(A)+1.
  \]
\end{proposition}
\begin{proof}[Proof sketch]
  Put $m:=\degree(A)$, and choose $m+1$ pairwise distinct names $a_1,\ldots, a_{m+1}$ such that $\supp(q_0)\seq \{a_1,\ldots, a_{m+1}\}$. 

\begin{enumerate} \item One first shows that for every finite bar string $\sigma_1\sigma_2\cdots \sigma_n\in \barA^*$ and every run
\[ q_0\xto{\sigma_1} q_1\xto{\sigma_2} q_2 \xto{\sigma_3} \cdots \xto{\sigma_n} q_n\]
in $A$ there exists $\sigma_1'\sigma_2'\cdots \sigma_n'\in \barA^*$ and a run
  \begin{equation}\label{eq:runparsketch}
    q_0 \xto{\sigma_1'} q_1' \xto{\sigma_2'} q_2' \xto{\sigma_3'} \cdots
    \xto{\sigma_n'} q_n'
  \end{equation}
such that (1) $q_i$ and $q_i'$ lie in the same orbit for $i=1,\ldots, n$, (2) $\sigma_1'\sigma_2'\cdots \sigma_n' =_\alpha \sigma_1\sigma_2\cdots \sigma_n$, and (3) $\Names{\sigma_1'\sigma_2'\cdots \sigma_n'}\seq \{a_1,\ldots, a_{m+1}\}$. The proof is by induction on $n$ and rests on the observation that for every state $q$ at least one of the names $a_1,\ldots, a_{m+1}$ is fresh for $q$.

\item Now suppose that the infinite bar string $w=\sigma_1\sigma_2\sigma_3\cdots\in \barA^\omega$ is accepted by $A$ via the accepting run
  \[
    q_0 \xto{\sigma_1} q_1 \xto{\sigma_2} q_2 \xto{\sigma_3} \cdots
  \]
 Consider the set of all
  partial runs \eqref{eq:runparsketch} satisfying the above conditions (1)--(3). This set organizes into a tree with the edge relation
  given by extension of runs; its nodes of depth $n$ are exactly
  the runs \eqref{eq:runparsketch}. By part 1 at least one
  such run exists for each $n\in \Nat$, i.e.~the tree is infinite. It
  is finitely branching because
  $\supp(q_i')\seq \{a_1,\ldots, a_{m+1}\}$ for all $i$ by \autoref{lem:rnna-props}, i.e.~there
  are only finitely many runs \eqref{eq:runparsketch} for each $n$. Hence, by
  K\H{o}nig's
  lemma the tree contains an infinite path. This yields an infinite run
  \[
    q_0 \xto{\sigma_1'} q_1' \xto{\sigma_2'} q_2' \xto{\sigma_3'}
    \cdots
  \]
  such that $q_i'$ and $q_i$ lie in the same orbit for each $i$ and
  $\sigma_1'\cdots \sigma_n'=_\alpha \sigma_1\cdots \sigma_n$ for each
  $n$. It follows that the Büchi RNNA $A$ accepts the infinite bar string
  \[
    w'= \sigma_1'\sigma_2'\sigma_3'\cdots.
  \]
 Moreover $w'=_\alpha w$ and
  $\supp(q_0)\cup \Names{w'}\seq \{a_1,\ldots, a_{m+1}\}$ as required.\qedhere
\end{enumerate}
\end{proof}

\section{Name-Dropping Modification}\label{sec:name-dropping}
Although the transitions of a Büchi RNNA are $\alpha$-invariant, its
literal $\omega$-language generally fails to be closed under
$\alpha$-equivalence. For instance, consider the Büchi RNNA with
states $\{q_0\}\cup \names \cup \names^2$ and transitions displayed
below, where $a,b$ range over distinct names in $\names$:
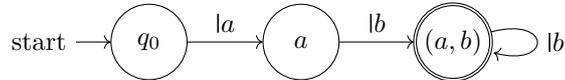
\begin{center}
  \tikzset{every state/.style={minimum size=20pt}}
  \begin{tikzpicture}[align=center,node distance=2cm, state/.style={circle, draw, minimum size=1cm, inner sep=1pt}]
    \node[state,initial] (q0) {$q_0$};
    \node[state, right of=q0] (a) {$a$};
    \node[state, right of=a, accepting] (ab) {$(a,b)$};
    \draw
    (q0) edge[above] node{$\newletter a$} (a)
    (a) edge[above] node{$\newletter b$} (ab)
    (ab) edge[loop right] node{$\newletter b$} (ab);
  \end{tikzpicture}
\end{center}
The automaton accepts the infinite bar string
$\newletter a\newletter b\newletter b\newletter b\newletter
b\newletter b\cdots $ for $a\neq b$ but does not accept the
$\alpha$-equivalent
$\newletter a \newletter a \newletter a \newletter a \newletter a
\newletter a\cdots$: the required transitions $a\xto{\scriptnew a} (a,b)$ and
$(a,b)\xto{\scriptnew a} (a,b)$ do not exist since $a$ is in the least support of the state $(a,b)$, which prevents $\alpha$-renaming of the given transitions $a\xto{\scriptnew b} (a,b)$ and $(a,b)\xto{\scriptnew b}(a,b)$.
A possible fix is to add a new
state ``$(\bot, b)$'' that essentially behaves like $(a,b)$ but has the name
$a$ dropped from its least support.

This idea can be generalized: We will show below how to transform any
Büchi RNNA $A$ into a Büchi RNNA $\tl{A}$ such that
$L_{0,\omega}(\tl{A})$ is the closure of $L_{0,\omega}(A)$ under
$\alpha$-equivalence, using the \emph{name-dropping modification}
(\autoref{def:ndc}). The latter simplifies a construction of the same
name previously given for RNNA over finite words~\cite{SchroderEA17},
and it is the key to the decidability of bar language inclusion proved
later on. First, some technical preparations:
\begin{rem}[Strong nominal sets]\label{rem:strong-nominal-sets}
  A nominal set $Y$ is called \emph{strong}~\cite{tzevelekos08} if for
  every $\pi\in \Perm(\names)$ and $y\in Y$ one has $\pi\o y = y$ if
  and only if $\pi(a)=a$ for all $a\in \supp(y)$. (Note that ``if''
  holds in all nominal sets.) As shown
  in~\cite[Lem.~2.4.2]{petrisan11} or~\cite[Cor.~B.27(2)]{mu19} strong
  nominal sets are up to isomorphism precisely the nominal sets of the
  form $Y=\sum_{i\in I} \names^{\# X_i}$ where $X_i$ is a finite set
  and $\names^{\# X}$ denotes the nominal set of all injective maps
  from~$X$ to~$\names$, with the group action defined pointwise;
  $\names^{\# X}$ may be seen as the set of possible configurations of
  an $X$-indexed set of registers containing pairwise distinct
  names. The set of orbits of~$Y$ is in bijection with~$I$; in
  particular,~$Y$ is orbit-finite iff~$I$ is finite. Strong nominal
  sets can be regarded as analogues of free algebras in categories of
  algebraic structures:
  \begin{enumerate}
  \item For every nominal set $Z$ there exists a strong nominal set $Y$ and a surjective equivariant map $e\colon Y\epito Z$ that is
    \emph{$\supp$-nondecreasing}, i.e.~$\supp(e(y))=\supp(y)$ for all
    $y\in Y$~\cite[Cor.~B.27.1]{mu19}. (Recall that
    $\supp(e(y))\seq \supp(y)$ always holds for equivariant maps.)
    Specifically, one can choose $Y=\sum_{i\in I} \names^{\# n_i}$
    where $I$ is the set of orbits of $Z$ and $n_i$ is the size of the
    support of any element of the orbit $i$. In particular, if~$Z$ is
    orbit-finite then~$Y$ is orbit-finite with the same number of
    orbits.
  
  \item Strong nominal sets are \emph{projective} w.r.t.~
    $\supp$-nondecreasing quotients~\cite[Lem.~B.28]{mu19}: given a
    strong nominal set $Y$, an equivariant map $h\colon Y\to Z$ and a
    $\supp$-nondecreasing quotient $e\colon X\epito Z$,
    there exists an equivariant map $g\colon Y\to X$ such that
    $h=e\o g$.
  \end{enumerate}
\end{rem}
\begin{proposition}\label{prop:strong}
  For every Büchi RNNA there exists a Büchi RNNA accepting the same
  literal~$\omega$-lan\-guage whose states form a strong
  nominal set.
\end{proposition}
\begin{proof}[Proof sketch]
 Given a Büchi RNNA viewed as a coalgebra
    \[
      Q\xto{\gamma} 2\times \Pow_\ufs(\names\times Q)\times
      \Pow_\ufs([\names]Q),
    \]
    express the nominal set $Q$ of states as a $\supp$-nondecreasing
    quotient $e\colon P\epito Q$ of an orbit-finite strong nominal set
    $P$ using \autoref{rem:strong-nominal-sets}. Note that the type
    functor
    $F(\dash) = 2\times \Pow_\ufs(\names\times \dash)\times
    \Pow_\ufs([\names](\dash))$ preserves $\supp$-nondecreasing quotients
    because all the functors
    $\Pow_\ufs(\dash), \names\times \dash$ and $[\names](\dash)$
    do. Therefore, the right vertical arrow in the square below is
    $\supp$-nondecreasing, so projectivity of $P$ yields an equivariant
    map $\beta$ making it commute:
\[
      \begin{tikzcd}
        P \ar[d,"e"', two heads] \ar[dashed]{r}{\beta}
        &
        2\times \Pow_\ufs(\names\times P)\times \Pow_\ufs([\names]P)
        \ar[d,"2\times \Pow_\ufs(\names\times e)\times
        \Pow_\ufs({[\names]}e)", two heads]
        \\
        Q \ar[r,"\gamma"]
        &
        2\times \Pow_\ufs(\names\times Q)\times \Pow_\ufs([\names]Q) 
      \end{tikzcd}
\]
    Thus $e$ is a coalgebra homomorphism from $(P,\beta)$ to
    $(Q,\gamma)$. It is not difficult to verify that for every $p\in P$ the states $p$ and $e(p)$ accept the same literal $\omega$-language. In particular, equipping $P$ with an initial state $p_0\in P$  such that $e(p_0)=q_0$ we see that the Büchi RNNAs $P$ and $Q$ accept the same literal $\omega$-language.
\end{proof}

\noindent A Büchi RNNA with states $Q=\sum_{i\in I} \names^{\# X_i}$ can be
interpreted as an automaton with a finite set $I$ of control states
each of which comes equipped with an $X_i$-indexed set of
registers. The following construction shows how to modify such an
automaton, preserving the accepted bar $\omega$-language, to become lossy in the
sense that after each transition some of the register contents may be
nondeterministically erased. Technically, this involves replacing~$Q$
with $\tl{Q}=\sum_{i\in I} \names^{\$ X_i}$ where $\names^{\$ X_i}$ is
a the nominal set of \emph{partial} injective maps from $X_i$ to
$\names$, again with the pointwise group action. (Note that while
$\names^{\# X_i}$ has only one orbit, $\names^{\$ X_i}$ has one orbit
for every subset of~$X_i$.) We represent elements of $\tl{Q}$ as pairs
$(i,r)$ where $i\in I$ and $r\colon X_i\to \names$ is a partial
injective map. The least support of $(i,r)$ is given by
\[
  \supp(i,r) = \supp(r) = \{ r(x) \compr x\in  \dom(r)\},
\]
where $\dom(r)$ is the \emph{domain} of $r$, i.e. the set of all
$x\in X_i$ for which $r(x)$ is defined. We say that a partial map
$\ol{r}\colon X_i\to \names$ \emph{extends} $r$ if
$\dom(r)\seq \dom(\ol{r})$ and $r(x)=\ol{r}(x)$ for all
$x\in \dom(r)$.

\begin{defn}[Name-dropping modification]\label{def:ndc}
  Let $A$ be a Büchi RNNA whose states form a strong nominal set
  $Q=\sum_{i\in I} \names^{\# X_i}$ (with~$I$ and all~$X_i$
  finite). The \emph{name-dropping modification} of $A$ is the Büchi RNNA
  $\widetilde{A}$ defined as follows:
  \begin{enumerate}
  \item The states are given by $\tl{Q} = \sum_{i\in I}\names^{\$ X_i}$.
  \item The initial state of $\tl{A}$ is equal to the initial state of $A$.
  \item $(i,r)$ is final in $\tl{A}$ iff some (equivalently,
    every\footnote{This is due to the equivariance of the subset $F
      \subseteq Q$ of final states.})
    state $(i,\dash)$ in $A$ is final.

  \item $(i,r)\xto{a} (j,s)$  in $\tl{A}$ iff $\supp(s)\cup \{a\}
    \seq \supp(r)$ and $(i,\ol{r})\xto{a} (j,\ol{s})$ in $A$ for some
    $\ol{r},\ol{s}$ extending $r, s$.
    
  \item $(i,r)\xto{\scriptnew a} (j,s)$ in $\tl{A}$ iff $\supp(s)\seq
    \supp(r)\cup \{a\}$ and $(i,\ol{r})\xto{\scriptnew b} (j,\ol{s})$ in $A$ for
    some $b\fresh s$ and some $\ol{r},\ol{s}$ extending $r,(a\,b)s$.
  \end{enumerate}
\end{defn}
\begin{theorem}\label{thm:ncd_buechi_rnna}
  For every Büchi RNNA $A$, the literal $\omega$-language
  of $\tl A$ is the closure of that of~$A$ under $\alpha$-equivalence.
\end{theorem}
\begin{proof}[Proof sketch]
  \begin{enumerate}
  \item\label{thm:ncd_buechi_rnna:1} One first verifies that the
    literal $\omega$-language of $\tl{A}$ is closed under
    $\alpha$-equi\-va\-lence. To this end, one proves more generally
    that given $\alpha$-equivalent infinite bar strings
    $w=\sigma_1\sigma_2\sigma_3\cdots$ and
    $w'=\sigma_1'\sigma_2'\sigma_3'\cdots$ and a run
    \[
      (i_0,r_0)\xto{\sigma_1} (i_1,r_1)\xto{\sigma_2}
      (i_2,r_2)\xto{\sigma_3} \cdots
    \]
    for $w$ in $\tl{A}$ there exists a run of the form
    \[
      (i_0,r_0')\xto{\sigma_1'} (i_1,r_1')\xto{\sigma_2'}
      (i_2,r_2')\xto{\sigma_3'} \cdots.
    \]
    for $w'$ in $\tl{A}$. In particular, by definition of the final
    states of $\tl{A}$, the first run is accepting iff the second one
    is. Similar to the proof of \autoref{lem:restrict_num_names}, one
    first establishes the corresponding statement for finite runs by
    induction on their length, and then extends to the infinite case
    via K\H{o}nig's lemma.
  
  \item In view of part \ref{thm:ncd_buechi_rnna:1} it remains to prove
    that $A$ and $\tl{A}$ accept the same bar
    $\omega$-language. Again, this follows from a more general
    observation: for every infinite run
    \[
      (i_0,r_0)\xto{\sigma_1} (i_1,r_1)\xto{\sigma_2}
      (i_2,r_2)\xto{\sigma_3} \cdots
    \]
    in $A$ there exists an infinite run of the form 
    \[
      (i_0,r_0')\xto{\sigma_1'} (i_1,r_1')\xto{\sigma_2'}
      (i_2,r_2')\xto{\sigma_3'} \cdots
    \]
    in $\tl{A}$ such that
    $\sigma_1'\sigma_2'\sigma_3'\cdots =_\alpha
    \sigma_1\sigma_2\sigma_3\cdots$. Moreover, the symmetric statement
    holds where the roles of $A$ and $\tl{A}$ are swapped. As before,
    one first shows the corresponding statement for finite bar strings
    and then invokes K\H{o}nig's lemma.\qedhere
  \end{enumerate}
\end{proof}
\section{Decidability of Inclusion}\label{sec:decidability}
With the preparations given in the previous sections, we
arrive at our main result: language inclusion of Büchi RNNA is
decidable, both under bar language semantics and data language
semantics. The main ingredients of our proofs below are the
name-dropping modification (\autoref{thm:ncd_buechi_rnna}) and the
name restriction property stated in
\autoref{lem:restrict_num_names}. We will employ them to show that the
language inclusion problems for Büchi RNNA reduce to the inclusion
problem for ordinary Büchi automata over finite alphabets. The latter is
well-known to be decidable with elementary complexity; in fact, it is $\PSPACE$-complete~\cite{kv96}.
\begin{rem}
  For algorithms deciding properties of Büchi RNNAs, a finite
  representation of the underlying nominal sets of states and
  transitions is required. A standard representation of a single orbit
  $X$ is given by picking an arbitrary element $x\in X$ with
  $\supp(x)=\{a_1,\ldots, a_m\}$ and forming the subgroup
  $G_X\seq \Perm(\{a_1,\ldots,a_m\})$ of all permutations
  $\pi\colon \{a_1,\ldots, a_m\}\to \{a_1,\ldots, a_m\}$ such that
  $\pi\o x = x$. The finite group $G_X$ determines $X$ up to
  isomorphism~\cite[Thm. 5.13]{Pitts13}. More generally, an
  orbit-finite nominal set consisting of the orbits $X_1,\ldots, X_n$
  is represented by a list of $n$ finite permutation groups
  $G_{X_1},\ldots, G_{X_n}$.
\end{rem}
We first consider the bar language semantics:
\begin{theorem}\label{thm:inclusion_decidable}
  Bar $\omega$-language inclusion of Büchi RNNAs is decidable in parametrized
  polynomial space.
\end{theorem}
\begin{proof}
  Let $A$ and $B$ be Büchi RNNAs; the task is to decide whether $L_{\alpha,\omega}(A)\seq L_{\alpha,\omega}(B)$. Put $m=\degree(A)$, and choose a set
  $S\seq \names$ of $m+1$ distinct names containing $\supp(q_{0})$, where~$q_0$ is the initial state of $A$. Moreover,
 form the finite alphabet
$
    \ol{S}=S\cup \{\newletter s\compr s\in S\}.
$ 
  \begin{enumerate}
  \item We claim that
    \begin{equation}\label{eq:acc_equiv} L_{\alpha,\omega}(A)\seq
      L_{\alpha,\omega}(B)\qquad\text{iff}\qquad L_{0,\omega}(A)\cap
      \ol{S}^\omega \seq L_{0,\omega}(\tl{B})\cap
      \ol{S}^\omega, \end{equation}
    where $\tl{B}$ is the name-dropping modification of $B$.

    \smallskip\noindent ($\Rightarrow$) Suppose that
    $L_{\alpha,w}(A)\seq L_{\alpha,\omega}(B)$. Then
    $L_{0,\omega}(A)\seq L_{0,\omega}(\tl{B})$ because
    $L_{0,\omega}(\tl{B})$ is the closure of $L_{0,\omega}(B)$ under
    $\alpha$-equivalence by \autoref{thm:ncd_buechi_rnna}. In
    particular,
    $L_{0,\omega}(A)\cap \ol{S}^\omega \seq L_{0,\omega}(\tl{B})\cap
    \ol{S}^\omega$.

    \smallskip\noindent ($\Leftarrow$) Suppose that
    $L_{0,\omega}(A)\cap \ol{S}^\omega \seq L_{0,\omega}(\tl{B})\cap
    \ol{S}^\omega$, and let $[w]_\alpha \in
    L_{\alpha,\omega}(A)$. Thus, $w=_\alpha v$ for some
    $v\in L_{0,\omega}(A)$. By \autoref{lem:restrict_num_names} we
    know that there exists $v'\in \ol{S}^\omega$ accepted by $A$ such
    that $v'=_\alpha v$. Then
    \[
      v'\in L_{0,\omega}(A)\cap \ol{S}^\omega \seq
      L_{0,\omega}(\tl{B})\cap \ol{S}^\omega \seq
      L_{0,\omega}(\tl{B}).
    \]
    Thus
    $[w]_\alpha=[v]_\alpha = [v']_\alpha \in
    L_{\alpha,\omega}(\tl{B})=L_{\alpha,\omega}(B)$, proving
    $L_{\alpha,\omega}(A)\seq L_{\alpha,\omega}(B)$.

  \item Now observe that both $L_{0,\omega}(A)\cap \ol{S}^\omega$ and
    $L_{0,\omega}(\tl{B})\cap \ol{S}^\omega$ are $\omega$-regular
    languages over the alphabet $\ol{S}$: they are accepted by the
    Büchi automata obtained by restricting the Büchi RNNAs $A$ and
    $B$, respectively, to the finite set of states with support
    contained in $S$ and transitions labeled by elements of
    $\ol{S}$. Inclusion of $\omega$-regular languages is decidable in
    polynomial space~\cite{kv96}. Let $k_A/k_B$ and $m_A/m_B$ denote
    the number of orbits and the degree of $A$/$B$. Since the support
    of every state of the Büchi automaton derived from $A$ is contained in
    the set $S$ and the latter has $(m_A+1)!$ permutations, there are at most
    \[
      k_A\cdot (m_A+1)!
      \in 
      \CO(k_A\cdot 2^{(m_A+1)\log (m_A+1)})
    \]
    states.
    The Büchi automaton derived from $\tl{B}$ has at most
    \[
      k_B\cdot 2^{m_B}\cdot (m_A+1)!
      \in
      \CO(k_B\cdot 2^{m_B + (m_A+1)\log (m_A+1)})
    \]
    states 
    since $\tl{B}$ has at most $k_B\cdot 2^{m_B}$ orbits. Thus, the
    space required for the inclusion check is polynomial in $k_A, k_B$
    and exponential in $m_B+(m_A+1)\log (m_A+1)$.\qedhere
  \end{enumerate}
\end{proof}

\noindent Next, we turn to the data language semantics. 
\begin{nota}
  Given infinite bar strings $w=\sigma_1\sigma_2\sigma_3\cdots$ and
  $w'=\sigma_1'\sigma_2'\sigma_3'\cdots$ we write $w\sqsubseteq w'$ if,
  for all $a\in \names$ and $n\in \Nat$,
  \[ \sigma_n=a
    \quad\text{implies}\quad
    \sigma_n'\in \{a,\newletter a\}
    \qquad\text{and}\qquad
    \sigma_n=\newletter a
    \quad\text{implies}\quad
    \sigma_n'=\newletter a.
  \]
  Thus, $w'$ arises from $w$ by arbitrarily putting bars in front of letters in $w$.
\end{nota}
\begin{lemma}\label{lem:order-rel}
  If $v=_\alpha w$ and $v\sqsubseteq v'$, then there exists
  $w'\in \barA^\omega$ such that $w\sqsubseteq w'$ and $v'=_\alpha w'$.
\end{lemma}
\begin{proof}
  Let $v'=\sigma_1'\sigma_2'\sigma_3'\cdots$ and
  $w=\rho_1\rho_2\rho_3\cdots$.  We define $w'$ to be the following
  modification of $w$: for every $n\in \Nat$, if $\rho_n=a$ and
  $\sigma_n'=\newletter b$ for some $a,b\in \names$, replace $\rho_n$
  by $\newletter a$. Then $w\sqsubseteq w'$ and $v'=_\alpha w'$ as
  required.
\end{proof}

\noindent In the following, we write $D(w)$ for $D(\{[w]_\alpha\})$; thus,
$D(w)$ is the set of all $\ub(v)\in \names^\omega$ where
$v=_\alpha w$ (cf.~\autoref{not:local-semantics}).
\begin{lemma}\label{lem:word-inclusion}
  Let $L$ be a bar $\omega$-language accepted by some Büchi RNNA and
  let $w\in \barA^\omega_\fs$. Then $D(w)\seq D(L)$ if and only if
  there exists $w'\sqsupseteq w$ such that $[w']_\alpha\in L$.
\end{lemma}
\begin{corollary}\label{cor:lang-inclusion}
  Let $K$ and $L$ be bar $\omega$-languages accepted by Büchi
  RNNA. Then $D(K)\seq D(L)$ if and only if for every
  $w\in \barA^\omega_\fs$ with $[w]_\alpha\in K$ there exists
  $w'\sqsupseteq w$ such that $[w']_\alpha\in L$.
\end{corollary}
\begin{proof}
  This follows from \autoref{lem:word-inclusion} using that
  $D(K)=\bigcup_{[w]_\alpha\in K} D(w)$ and that every
  $\alpha$-equivalence class $[w]_\alpha\in K$ has a finitely
  supported representative by \autoref{lem:restrict_num_names}.
\end{proof}
\begin{theorem}\label{thm:inclusion_decidable_local}
  Data $\omega$-language inclusion of Büchi RNNAs is decidable in parametrized
  polynomial space.
\end{theorem}
\begin{proof}
  Let $A$ and $B$ be Büchi RNNAs; the task is to decide whether
  $D_\omega(A)\seq D_\omega(B)$. Put $m=\degree(A)$, and choose a set $S\seq \names$ of $m+1$
  distinct names containing $\supp(q_{0})$, where $q_0$ is the initial state of $A$. Moreover, form the finite
  alphabet
$
    \ol{S}=S\cup \{\newletter s\compr s\in S\}.
$
  \begin{enumerate}
  \item For every language $L\seq \ol{S}^\omega$ let $L{\downarrow}$
    denote the downward closure of $L$ with respect to $\sqsubseteq$:
    \[
      L{\downarrow} = \{ v\in \ol{S}^\omega\compr \text{there exists $w\in
        L$ such that } v\sqsubseteq w \}.
    \]
    Every Büchi automaton accepting $L$ can be turned into a Büchi
    automaton accepting~$L{\downarrow}$ by adding the transition
    $q\xto{a}q'$ for every transition $q\xto{\scriptnew a} q'$ where
    $a\in S$.
    
  \item We claim that
    \begin{equation}\label{eq:acc_equiv_local}
      D_\omega(A)\seq D_\omega(B)
      \qquad\text{iff}\qquad
      L_{0,\omega}(A)\cap \ol{S}^\omega \seq (L_{0,\omega}(\tl{B})\cap
      \ol{S}^\omega){\downarrow},
    \end{equation}
    where $\tl{B}$ is the name-dropping modification of $B$.

    \smallskip\noindent ($\Rightarrow$) Suppose that
    $D_\omega(A)\seq D_\omega(B)$, that is
    $D(L_{\alpha,w}(A))\seq D(L_{\alpha,\omega}(B))$, and let
    $w\in L_{0,\omega}(A)\cap \ol{S}^\omega$. Then
    $[w]_\alpha\in L_{\alpha,\omega}(A)$, hence by
    \autoref{cor:lang-inclusion} there exists $w'\sqsupseteq w$ such
    that $[w']_\alpha\in L_{\alpha,\omega}(B)$. Then
    $w'\in L_{0,\omega}(\tl{B})\cap \ol{S}^\omega$ because
    $L_{0,\omega}(\tl{B})$ is the closure of $L_{0,\omega}(B)$ under
    $\alpha$-equivalence by \autoref{thm:ncd_buechi_rnna}. Thus
    $w\in (L_{0,\omega}(\tl{B})\cap \ol{S}^\omega){\downarrow}$, which
    proves
    $ L_{0,\omega}(A)\cap \ol{S}^\omega \seq (L_{0,\omega}(\tl{B})\cap
    \ol{S}^\omega){\downarrow}$.

    \smallskip\noindent ($\Leftarrow$) Suppose that
    $L_{0,\omega}(A)\cap \ol{S}^\omega \seq (L_{0,\omega}(\tl{B})\cap
    \ol{S}^\omega){\downarrow}$. To prove
    $D_\omega(A)\seq D_\omega(B)$, that
    is~$D(L_{\alpha,\omega}(A))\seq D(L_{\alpha,\omega}(B))$, we use
    \autoref{cor:lang-inclusion}. Thus, let
    $[w]_\alpha \in L_{\alpha,\omega}(A)$. Then $w=_\alpha v$ for some
    $v\in L_{0,\omega}(A)$. By \autoref{lem:restrict_num_names} we may
    assume that $v\in \ol{S}^\omega$.  By hypothesis this implies
    $v\in (L_{0,\omega}(\tl{B})\cap \ol{S}^\omega){\downarrow}$, that
    is, there exists $v'\sqsupseteq v$ such that
    $v'\in L_{0,\omega}(\tl{B})$. By \autoref{lem:order-rel}, there
    exists $w'\sqsupseteq w$ such that $w'=_\alpha v'$. Then
    $[w']_\alpha = [v']_\alpha \in L_{\alpha,\omega}(\tl{B}) = 
    L_{\alpha,\omega}(B)$, as required.

  \item The decidability of $D_\omega(A)\seq D_\omega(B)$ now follows
    from part~1 and \eqref{eq:acc_equiv_local} in complete analogy to the proof
    of \autoref{thm:inclusion_decidable}.  \qedhere
  \end{enumerate}
\end{proof}
\begin{rem}\label{rem:data-language-finite-support}
  In the above theorem it is crucial to admit non-finitely supported
  data words: it is an open problem whether the inclusion
  $D_\omega(A)\cap \names^\omega_\fs\seq D_\omega(B)\cap
  \names^\omega_\fs$ is decidable. In fact, our decidability proof
  relies on \autoref{cor:lang-inclusion} as a key ingredient, and the
  latter fails if the condition $D(K)\seq D(L)$ is replaced by the
  weaker condition
  $D(K)\cap \names^\omega_\fs \seq D(L)\cap \names^\omega_\fs$.

  To see this, let $K$ and $L$ be the bar $\omega$-languages accepted
  by the two Büchi RNNAs displayed below, where $a,b$ range over names
  in $\names$ and $a\neq b$:
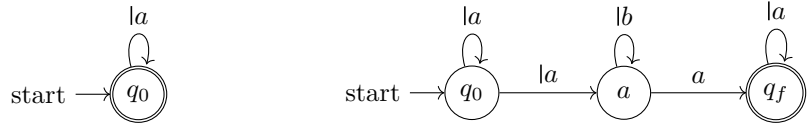
  \begin{center}
    \tikzset{every state/.style={minimum size=20pt}}
    \begin{tikzpicture}[align=center,node distance=2cm] 
      \node[state, initial, accepting] (q0) {$q_0$};
      \draw (q0) edge[loop above] node{$\newletter a$} (q0);
    \end{tikzpicture}
    \qquad\qquad\qquad
    \begin{tikzpicture}[align=center,node distance=2cm] 
      \node[state,initial] (q0) {$q_0$};
      \node[state, right of=q0] (a) {$a$};
      \node[state, right of=a, accepting] (qf) {$q_f$};
      \draw
      (q0) edge[loop above] node{$\newletter a$} (q0)
      (q0) edge[above] node{$\newletter a$} (a)
      (a) edge[loop above] node{$\newletter b$} (a)
      (a) edge[above] node{$a$} (qf)
      (qf) edge[loop above] node{$\newletter a$} (qf);
    \end{tikzpicture}
  \end{center}
  We have $D(K)=\names^\omega$ and $D(L)$ consists of all data words
  in $\names^\omega$ in which some name occurs at least twice. Thus,
  $D(K)\cap \names^\omega_\fs = D(L)\cap \names^\omega_\fs =
  \names^\omega_\fs$, in particular the inclusion
  $D(K)\cap \names^\omega_\fs \seq D(L)\cap \names^\omega_\fs$
  holds. Consider the infinite bar string
  $w=\newletter a\newletter a\newletter a\cdots$ and note that
  $w'\sqsupseteq w$ implies $w'=w$. Then $[w]_\alpha \in K$ but
  $[w]_\alpha\not\in L$ since every bar string accepted by the
  right-hand automaton contains some letter from $\names$. 
  Thus, ``only if'' fails in \autoref{cor:lang-inclusion}.
\end{rem}
\section{Relation to Other Automata Models}\label{sec:relation}
We conclude this paper by comparing Büchi RNNA with two
related automata models over infinite words.  Ciancia and Sammartino
\cite{CianciaSammartino14} consider deterministic nominal automata
accepting data $\omega$-lan\-gu\-ages $L\seq \names^\omega$ via a
Muller acceptance condition. More precisely, a \emph{nominal
  deterministic Muller automaton} ($\emph{nDMA}$)
$A=(Q,\delta,q_0,\F)$ is given by an orbit-finite nominal set $Q$ of
states, and equivariant map $\delta\colon Q\times \names \to Q$
representing transitions, an initial state $q_0\in Q$, and a set
$\F\seq \Pow(\orb(Q))$ where $\orb(Q)$ is the finite set of orbits
of~$Q$. Every input word $w=a_1a_2a_3\cdots \in \names^\omega$ has a
unique run $q_0 \xto{a_1} q_1 \xto{a_2} q_2 \xto{a_3}\cdots $ where
$q_{i+1}=\delta(q_{i},a_{i+1})$ for $i=0,1,2,\ldots$. The word $w$ is
\emph{accepted} if the set $\{ O\in \orb(Q) \compr \text{$q_n\in O$ for
  infinitely many $n$} \}$ lies in $\F$. The data $\omega$-language
\emph{accepted} by the automaton is the set of all words
$w\in \names^\omega$ whose run is accepting.

As for Büchi RNNA, language inclusion is decidable for
nDMA~\cite[Thm. 4]{CianciaSammartino14}. In terms of expressive power the two models
are incomparable, as witnessed by the data $\omega$-languages
\begin{align*}
K &= \{ w\in \names^\omega\compr \text{some $a\in \names$ occurs  infinitely often in $w$} \},\\
L & =\{ w\in \names^\omega\compr \text{$w$ does not have  the suffix
        $a^\omega$ for any $a\in \names$} \}.
\end{align*}
\vspace{-0.6cm}
\begin{proposition}
  \begin{enumerate}
  \item The language $K$ is accepted by a Büchi RNNA but not by any nDMA.
    
  \item The language $L$ is accepted by an nDMA but not by any Büchi RNNA.
  \end{enumerate}
\end{proposition}
\begin{proof}
  \begin{enumerate}
  \item We have seen in \autoref{ex:rnna} that the language $K$ is
    accepted by a Büchi RNNA. We claim that $K$ is not accepted by any
    nDMA. Since the class of languages accepted by nDMA is closed
    under complement, it suffices to show that the
    language
    \[\ol{K}=\{ w\in \names^\omega\compr \text{each $a\in \names$ occurs
        only finitely often in $w$} \}\] is not accepted by any
    nDMA. Suppose towards a contradiction that $A=(Q,\delta,q_0,\F)$ is an
    nDMA accepting $\ol{K}$. Let $m$ be the maximum of all
    $|\supp(q)|$ where $q\in Q$.  Fix $m+1$ pairwise distinct names
    $a_1,\ldots, a_{m+1}\in \names$ and an arbitrary word
    $w_0\in \ol{K}$.

    Choose a factorization $w_0=v_1w_1$
    ($v_1\in \names^*, w_1\in \names^\omega)$ such that all
    occurrences of $a_1,\ldots, a_{m+1}$ in $w_0$ lie in the finite
    prefix $v_1$. Let $q_1$ be the state reached from $q_0$ on
    input~$v_1$.  Then $\supp(q_1)$ does not contain all of the names
    $a_1,\ldots, a_{m+1}$, say $a_i\fresh q_1$. Choose any name
    $a\in \names$ occurring in $w_1$ such that $a\# q_1$. Then
    $q_1=(a\, a_i)q_1$ accepts $(a\, a_i)w_1$ since by
    equivariance the run of $(a\, a_i)w_1$ from $q_1$ visits the
    same orbits as the run of $w_1$.

    Now repeat the same process with $q_1$ and $(a\,a_i)w_1$ in lieu
    of $q_0$ and $w_0$, and let \mbox{$(a\, a_i)w_1 = v_2w_2$} denote the
    corresponding factorization; note that $v_2$ is nonempty because
    $(a\,a_i)w_1$ contains the name $a_i$. Continuing in this fashion
    yields an infinite word $v=v_1v_2v_3\ldots$ such that each $v_i$
    $(i>1)$ contains at least one of the letters
    $a_1,\ldots, a_{m+1}$, and the run of~$v$ traverses the same
    orbits (in the same order) as the run of $w$. Thus $v$ is accepted
    by the nDMA $A$ although $v\not\in \ol{K}$, a contradiction.
    
  \item The language $L$ is accepted by the nDMA with states
    $\{q_0\} \cup \names\times \{0,1\}$ and transitions as displayed
    below, where $a,b$ range over distinct names in $\names$. The
    acceptance condition is given by
    $\F=\{\{\names\times \{0\}, \names\times \{1\}\}\}$.
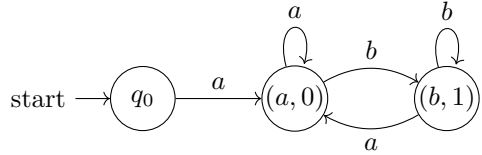
    \begin{center}
      \begin{tikzpicture}[align=center,node distance=2cm, state/.style={circle, draw, minimum size=.85cm, inner sep=0pt}] 
        \node[state,initial] (q0) {$q_0$};
        \node[state, right of=q0] (a0) {$(a,0)$};
        \node[state, right of=a0] (b1) {$(b,1)$};
        \draw
        (q0) edge[above] node{$a$} (a0)
        (a0) edge[loop above] node{$a$} (a0)
        (a0) edge[above, bend left] node{$b$} (b1)
        (b1) edge[loop above] node{$b$} (b1)
        (b1) edge[below, bend left] node{$a$} (a0);
      \end{tikzpicture}
    \end{center}
    We claim that $L$ is not accepted by any Büchi RNNA. Suppose to
    the contrary that $A=(Q,R,q_0,F)$ is a Büchi RNNA with
    $D_\omega(A)=L$. By \autoref{thm:ncd_buechi_rnna} we may assume
    that $L_{0,\omega}(A)$ is closed under $\alpha$-equivalence.
    Fix an arbitrary word $w=a_1a_2a_3\cdots$ whose names are pairwise
    distinct and not contained in $\supp(q_0)$. Then $w\in L$, so
    there exists $v\in L_{0,\omega}(A)$ such that $\ub(v)=w$. We claim
    that $v=\newletter a_1\newletter a_2\newletter a_3\cdots$. Indeed,
    if the~$n$th letter of $v$ is $a_n$, then in an accepting run for
    $v$ in $A$ the state $q$ reached before reading~$a_n$ must have
    $a_n\in \supp(q)$ by
    \autoref{lem:rnna-props}.\ref{lem:rnna-props:3}. But this is
    impossible because
    $\supp(q)\seq \supp(q_0)\cup \{a_1,\ldots, a_{n-1}\}$ again by
    \autoref{lem:rnna-props}.

    Thus $v=_\alpha \newletter a\newletter a\newletter a\cdots$. This
    implies $a^\omega\in D_\omega(A)$ although $a^\omega\not\in L$, a
    contradiction.\qedhere
  \end{enumerate}
\end{proof}

Let us note that the above result does not originate in weakness of the Büchi acceptance condition. One may generalize Büchi RNNA to \emph{Muller RNNA} where the final states $F\seq Q$ are replaced
by a set $\F\seq \Pow(\orb(Q))$, and a bar string $w\in \barA^\omega$ is said to be
\emph{accepted} if there exists a run for $w$ such
that the set of orbits visited infinitely often lies in $\F$. However, as for
classical nondeterministic Büchi and Muller automata, this does 
not affect expressivity:
\begin{proposition}\label{prop:muller-buchi-equiv}
  A literal $\omega$-language is accepted by a Büchi RNNA if and only
  if it is accepted by a Muller RNNA.
\end{proposition}
Finally, we mention a tight connection between Büchi RNNA and
\emph{session automata}~\cite{BolligEA14}. The data
($\omega$-)language associated to a (Büchi) RNNA uses a \emph{local
  freshness} semantics in the sense that its definition considers
possibly non-clean bar strings accepted by $A$. In some applications,
e.g.~nonce generation, a more suitable semantics is given by
\emph{global freshness} where only clean bar strings are admitted,
i.e.~one associates to $A$ the data languages
\begin{align*}
  D_{\#}(A)
  &= \{ \ub(w) \compr \text{$w\in \barA^*$ clean and $[w]_\alpha\in
    L_{\alpha}(A)$} \},
  \\
  D_{\omega,\#}(A)
  &= \{ \ub(w) \compr \text{$w\in \barA^\omega$ clean and $[w]_\alpha\in L_{\alpha,\omega}(A)$} \}. 
\end{align*}
For instance, for the Büchi RNNA $A$ from \autoref{ex:rnna} the
language $D_{\omega,\#}(A)$ consists of all infinite words
$w\in \names^\omega$ where exactly one name $a\in \names$ occurs
infinitely often and every name~$b\neq a$ occurs at most once.

Under global freshness semantics, it has been observed in previous
work~\cite{SchroderEA17} that a data language $L\seq \names^*$ is
accepted by some session automaton iff $L=D_{\#}(A)$ for some RNNA $A$
whose initial state $q_0$ has empty support. An analogous correspondence holds for Büchi
RNNA and data $\omega$-languages $L\seq \names^\omega$ if the original
notion of session automata is generalized to infinite words with a
Büchi acceptance condition.
\section{Conclusions and Future Work}

\noindent We have introduced \emph{B\"uchi regular nondeterministic
  nominal automata} (B\"uchi RNNAs), an automaton model for languages
of infinite words over infinite alphabets. B\"uchi RNNAs allow for
inclusion checking in elementary complexity (parametrized
polynomial space) despite the fact that they feature full
nondeterminism and do not restrict the number of registers
(contrastingly, even for register automata over finite
words~\cite{KaminskiFrancez94}, inclusion checking becomes decidable
only if the number of registers is bounded to at most~$2$).

 An
important further step will be to establish a logic-automata
correspondence of B\"uchi RNNAs with a suitable form of linear
temporal logic on infinite data words.

A natural direction for generalization is to investigate RNNAs over infinite trees with binders, modeled as coalgebras of type $\Pow_\ufs F$ for a functor $F$ associated to a binding signature and equipped with the ensuing notion of $\alpha$-equivalence due to Kurz et al.~\cite{kpsv13}. 

Finally, we would like to explore the bar language and data language semantics of Büchi RNNA from the perspective of coalgebraic trace semantics~\cite{ush16}, where infinite behaviours emerge as solutions of (nested) fixed point equations in Kleisli categories.

\bibliographystyle{plainurl}
\bibliography{coalgml}

\clearpage
\appendix

\section*{Appendix: Omitted Proofs}
This appendix provides all proofs omitted for space reasons.

\section{Details for \autoref{sec:alpha-equiv}}

\section*{Proof of \autoref{lem:alpha-equiv-coalgebraic}}
  The functor $F=\names\times (\dash) + [\names](\dash)$ preserves
  limits of $\omega^\op$-chains since such limits commute with
  $\times$ and $+$ (like in $\Set$) and the functor
  $[\names](\dash)$ is a right adjoint~\cite[Thm. 4.12]{Pitts13}, i.e.~preserves all limits. Therefore, the final
  coalgebra $\nu F$ arises as the limit of the final
  $\omega^\op$-chain of $F$:
  \[
    \nu F = \lim_{n\in \Nat} F^n 1.
  \]
  We prove by induction that
  \[
    F^n1 \cong \barA^n/{=_\alpha}\qquad\text{for every $n\in \Nat$},
  \]
  where $=_\alpha$ is $\alpha$-equivalence restricted to bar strings
  of length $n$.  The claim is obvious for $n=0$. For the induction
  step $n\to n+1$ we compute
  \[
    F^{n+1}1
    =
    F(F^n1)\cong F(\barA^n/{=_\alpha})
    =
    \names\times (\barA^n/{=_\alpha}) + [\names](\barA^n/{=_\alpha})
    \cong \barA^{n+1}/{=_\alpha}.
  \]
  Here, the first isomorphism uses the induction hypothesis and the
  second one is given by
  \[
    (a,[w]_\alpha) \mapsto  [aw]_\alpha
    \qquad\text{and}\qquad
    \braket{a} [w]_\alpha \mapsto [\newletter aw]_\alpha
    \qquad\text{for $a\in \names$ and $w\in \barA^n$.}
  \]
  This is clearly well-defined 
  according to the definition of $\alpha$-equivalence for finite
  words.

  Finally, we observe that the canonical map
  $e_{\alpha,n}\colon \nu G\to F^n1 \cong \barA^n/{=_\alpha}$ sends
  $w\in \barA^\omega_\fs$ to the $\alpha$-equivalence class of
  $w_n$ Since $e_\alpha$ merges two words iff each $e_{\alpha,n}$
  does, we are done. 

\section{Details for \autoref{sec:rnna}}

\section*{Details for \autoref{lem:restrict_num_names}}

  Put $m:=\degree(A)$. 

\begin{enumerate} 
\item Two finite runs \[ r_0\xto{\sigma_1} r_1\xto{\sigma_2} r_2\xto{\sigma_3}\cdots \xto{\sigma_n} r_n\quad\text{and}\quad r_0'\xto{\sigma_1'} r_1'\xto{\sigma_2'} r_2'\xto{\sigma_3'}\cdots \xto{\sigma_n'} r_n'  \]
in $A$ of the same length are called \emph{equivalent} if for each $i=0,\ldots, n$ the states $r_i$ and $r_i'$ lie in the same orbit. 
We claim that

\medskip\noindent \textbf{($\ast$)} For every $q\in Q$ and $w\in \barA^*$, if there exists a run of $w$ starting in $q$ then there exists an equivalent run starting in $q$ whose label $w'\in \barA^*$ satisfies
  \[
    w'=_\alpha w\qquad\text{and}\qquad |\supp (q) \cup \Names{w'}|\leq
    \degree(A)+1.
  \]
  The proof of ($\ast$) is by induction on the length of $w$. The base case
  $w=\epsilon$ is obvious.

  \noindent For the induction step, suppose $w=\sigma v$
  where $\sigma\in \barA$ and $v\in \barA^*$, and let $q\xto{\sigma}r$ be the
   first transition of a run of $w$ from $q$. Then $v$ has a run from $r$, so by induction there exists $v'\in \barA^*$ with an equivalent run from $r$ such 
  that $v'=_\alpha v$ and $|\supp(r)\cup \Names{v'}|\leq m+1$. Choose
  $m+1$ pairwise distinct names $a_1,\ldots,a_{m+1}$ such
  that \[\supp(r)\cup \Names{v'}\seq \{a_1,\ldots, a_{m+1}\}.\] We have $\sigma=a$ or $\sigma=\newletter a$ for some $a\in \names$ and consider the two cases separately:
  
  \smallskip\noindent\textbf{Case 1: $\sigma=a$}. Choose $m+1$ pairwise
  distinct names $b_1,\ldots,b_{m+1}$ such that
  \[
    \supp(q)\seq \{b_1,\ldots,b_{m+1}\},
  \]
  Since $q\xto{a} r$ we have $\supp(r)\seq \supp(q)$ by
  \autoref{lem:rnna-props}.\ref{lem:rnna-props:1}. Extend the inclusion map
  $\supp(r)\monoto \supp(q)$ to a bijection
  $\pi\colon \{a_1,\ldots,a_{m+1}\}\to \{b_1,\ldots,b_{m+1}\}$. Then $\pi\o v'$ has a run from
  $\pi\cdot r = r$ by equivariance, and
  $\pi\cdot v'=_\alpha v'$ because $\pi$ fixes all free names of $v'$
  (which are contained in $\supp(r)$ by \autoref{lem:rnna-props}.\ref{lem:rnna-props:3}). Put
  \[
    w' := a(\pi\cdot v').
  \]
  Then $w'$ has a run from $q$ equivalent to the given run of $w$: compose the transition $q\xto{a} r$ with the above run of $\pi\o v'$ from $r$. Moreover \[w'=_\alpha av'=_\alpha av = w\] and
  \[
    \supp(q)\cup \Names{w'} = \supp(q)\cup \Names{\pi\cdot v'}\cup
    \{a\}\seq \{b_1,\ldots, b_{m+1}\}
  \]
  using that $a\in \supp(q)$ because $q\xto{a} r$.
  
  \smallskip\noindent \textbf{Case 2: $\sigma = \newletter a$}. If $a\in \{a_1,\ldots,a_{m+1}\}$ one can choose $w':=\newletter av'$: composing the transition $q\xto{\newletter a} r$ with the given run of $v'$ from $r$ yields a run of $w'$ with the required properties. Thus suppose $a\not\in  \{a_1,\ldots,a_{m+1}\}$. Then $a\fresh r$ because $\supp(r)\seq \{a_1,\ldots,a_{m+1}\}$. Moreover, $a_i\fresh r$ for some $i$ because $|\supp(r)|\leq m$. We claim that the bar string
\[ w':= \newletter a_iv' \]
has the required properties. Indeed, we have the transition $q\xto{\scriptnew a_i} r$ in $A$ by $\alpha$-invariance because $q\xto{\scriptnew a} r$ and $\braket{a}r=\braket{a_i}r$. Composing this transition with the given run of $v'$ from $r$ yields a run of $w'$ from $q$ equivalent to the given run of $w$. Finally, note that the names $a$ and $a_i$ are not free in $v'$ because  $\FN(v')\seq \supp(r)$ by \autoref{lem:rnna-props}.\ref{lem:rnna-props:3}, so $v'=_\alpha (a\, a_i)v'$. This implies
\[ w'= \newletter a_iv' =_\alpha \newletter a (a\, a_i)v'=_\alpha \newletter av' =_\alpha \newletter av = w. \]
%
\item We are ready to prove the proposition. Let $w=\sigma_1\sigma_2\sigma_3\cdots$ and
  suppose that
  \[
    q_0 \xto{\sigma_1} q_1 \xto{\sigma_2} q_2 \xto{\sigma_3} \cdots
  \]
  is an accepting run for $w$. Choose $m+1$ pairwise distinct names
  $a_1,\ldots,a_{m+1}$ such that
  $\supp(q_0)\seq \{a_1,\ldots, a_{m+1}\}$. Consider the set of all
  partial runs
  \begin{equation}\label{eq:runpar}
    q_0 \xto{\sigma_1'} q_1' \xto{\sigma_2'} q_2' \xto{\sigma_3'} \cdots
    \xto{\sigma_n'} q_n'
  \end{equation}
  where (1)~$n\in \Nat$, (2)~$q_i'$ and $q_i$ lie in the same orbit
  for $i=1,\ldots, n$, (3)~$\sigma_1'\cdots \sigma_n' =_\alpha \sigma_1\cdots \sigma_n$, and
  (4)~$\Names{\sigma_1'\cdots \sigma_n'}\seq \{a_1,\ldots
  a_{m+1}\}$. This set organizes into a tree with the edge relation
  given by extension of runs. Thus, the nodes of depth $n$ are exactly
  the runs \eqref{eq:runpar}. By ($\ast$) in part 1 of the proof, at least one
  such run exist for each $n\in \Nat$, i.e.~the tree is infinite. It
  is finitely branching because
  $\supp(q_i')\seq \{a_1,\ldots, a_{m+1}\}$ for all $i$ by \autoref{lem:rnna-props}, i.e.~there
  are only finitely many runs \eqref{eq:runpar} for each $n$. By
  K\H{o}nig's
  lemma, the tree contains an infinite path, which yields an infinite run
  \[
    q_0 \xto{\sigma_1'} q_1' \xto{\sigma_2'} q_2' \xto{\sigma_3'}
    \cdots
  \]
  such that $q_i'$ and $q_i$ lie in the same orbit for each $i$ and
  $\sigma_1'\cdots \sigma_n'=_\alpha \sigma_1\cdots \sigma_n$ for each
  $n$. Thus, putting
  \[
    w'= \sigma_1'\sigma_2'\sigma_3'\cdots
  \]
  we see that $A$ accepts $w'$. Moreover, $w'=_\alpha w$ and
  $\supp(q_0)\cup \Names{w'}\seq \{a_1,\ldots, a_{m+1}\}$. 
\end{enumerate}

\section{Details for \autoref{sec:name-dropping}}

\section*{Proof of \autoref{prop:strong}}
  \begin{enumerate}
  \item Given a Büchi RNNA viewed as a coalgebra
    \[
      Q\xto{\gamma} 2\times \Pow_\ufs(\names\times Q)\times
      \Pow_\ufs([\names]Q),
    \]
    express the nominal set $Q$ of states as a $\supp$-nondecreasing
    quotient $e\colon P\epito Q$ of an orbit-finite strong nominal set
    $P$ using \autoref{rem:strong-nominal-sets}. Note that the type
    functor
    $F(\dash) = 2\times \Pow_\ufs(\names\times \dash)\times
    \Pow_\ufs([\names](\dash))$ preserves $\supp$-nondecreasing quotients
    because the functors
    $\Pow_\ufs(\dash), \names\times \dash$ and $[\names](\dash)$
    do. Therefore, the right vertical arrow in the diagram below is
    $\supp$-nondecreasing, and projectivity of $P$ yields an equivariant
    map $\beta$ making the diagram commute.
    \begin{equation}\label{eq:coalg-morphism}
      \begin{tikzcd}
        P \ar[d,"e"', two heads] \ar[dashed]{r}{\beta}
        &
        2\times \Pow_\ufs(\names\times P)\times \Pow_\ufs([\names]P)
        \ar[d,"2\times \Pow_\ufs(\names\times e)\times
        \Pow_\ufs({[\names]}e)", two heads]
        \\
        Q \ar[r,"\gamma"]
        &
        2\times \Pow_\ufs(\names\times Q)\times \Pow_\ufs([\names]Q) 
      \end{tikzcd}
    \end{equation}
    Thus $e$ is a coalgebra homomorphism from $(P,\beta)$ to
    $(Q,\gamma)$. The latter implies that the following three
    properties hold for each $p\in P$ and $\sigma\in \barA$:
    \begin{enumerate}[(1)]
    \item\label{hom1} If $p\xto{\sigma} r$ in $P$ then $e(p)\xto{\sigma} e(r)$ in $Q$.
    \item\label{hom2} If $e(p)\xto{\sigma} q$ in $Q$ then there exists
      $r\in P$ such that $p\xto{\sigma} r$ in $P$ and $e(r)=q$.
    \item\label{hom3} The state $p$ is final in $P$ iff $e(p)$ is final in $Q$.
    \end{enumerate}
    Indeed, clearly commutativity of the left component of
    \eqref{eq:coalg-morphism} yields~(3), and commutativity of the
    central component yields~(1) and~(2) for
    $\sigma\in \names$. It remains to show that commutativity of the right
    component yields~(1) and~(2) for $\sigma=\newletter a\in \barA$. Let
    \[\beta_r\colon P\to \Pow_\ufs([\names]P)\qquad\text{and}\qquad \gamma_r\colon Q\to \Pow_\ufs([\names]Q)\] denote the right
    components of $\beta$ and~$\gamma$.

    \smallskip\noindent Ad~(1). If $p\xto{\scriptnew a} r$ in $P$ we have
    $\braket{a} r\in \beta_r(p)$. Thus,
    $\braket{a} e(r) = [\names]e(\braket{a} r)\in \gamma_r(e(p))$
    by~\eqref{eq:coalg-morphism}. This shows $e(p)\xto{\scriptnew a} e(r)$ in
    $Q$.

    \smallskip\noindent Ad~(2). If $e(p)\xto{\scriptnew a} q$ in $Q$ we have
    $\braket{a}q\in \gamma_r(e(p))$. From \eqref{eq:coalg-morphism} it
    follows that there exists $b\in \names$ and $r\in P$ such that
    $\braket{b}r\in \beta_r(p)$ and $\braket{a}q = \braket{b}e(r)$. If
    $a=b$, this implies $p\xto{\scriptnew a}r$ in $P$ and $q=e(r)$, and we are
    done. If $a\neq b$, we have $a\fresh e(r)$ because
    $\supp(q)\setminus \{a\} = \supp(e(r))\setminus \{b\}$. Since $e$
    is $\supp$-nondecreasing, this implies $a\fresh r$. Therefore
    $\braket{a}(a\, b)r = \braket{b} r \in \beta_r(p)$,
    i.e.~$p\xto{\scriptnew a}(a\, b)r$ in $P$. Moreover, we have
    $\braket{a}q = \braket{b}e(r) = \braket{a}(a\, b)e(r) =
    \braket{a}e((a\,b)r)$ and therefore $e((a\,b)r)=q$, as required.

  \item Choose $p_0\in P$ such that $e(p_0)=q_0$. The coalgebra
    $(P,\beta)$ can be viewed as a Büchi RNNA equipped with the initial
    state $p_0$. We claim that the Büchi RNNAs $P$ and $Q$ accept the same
    literal ($\omega$-)language. In fact, this is immediate from part
    1: by~(1) and~(2) we see that for every run
    \[
      p_0 \xto{\sigma_1} p_1 \xto{\sigma_2} p_2 \xto{\sigma_3} \cdots
    \]
    in $P$ there exists a corresponding run
    \[
      q_0 \xto{\sigma_1} q_1 \xto{\sigma_2} q_2 \xto{\sigma_3} \cdots
    \]
    in $Q$ with $e(p_i)=q_i$ for all $i$, and vice versa. Moreover, by
    (3) the first run above is accepting iff the second one
    is. This concludes the proof. \qedhere
  \end{enumerate}

\section*{Details for \autoref{thm:ncd_buechi_rnna}}

\begin{lemma}
  The automaton $\tl{A}$ is an RNNA.
\end{lemma}
\begin{proof}
  \begin{enumerate}
  \item Transitions are equivariant: Let $\pi\in \Perm(\names)$ and
    $a\in \names$.

    Suppose that $(i,r)\xto{a} (j,s)$ in $\tl{A}$. Then
    $(i,\ol{r})\xto{a} (j,\ol{s})$ in $A$ for some $\ol{r}, \ol{s}$
    extending $r,s$. By equivariance of transitions in $A$ we get
    $(i,\pi\cdot \ol{r})\xto{\pi(a)} (j,\pi\cdot \ol{s})$ in $A$. Then
    $\pi\cdot \ol{r},\pi\cdot \ol{s}$ extends
    $\pi\cdot {r}, \pi \cdot s$ and moreover
    \[
      \supp(\pi\cdot s)\cup \{\pi(a)\}
      =
      \pi(\supp(s)\cup \{a\})
      \seq
      \pi(\supp(r))
      =
      \supp(\pi\cdot r).
    \]
    Thus $(i,\pi\cdot {r})\xto{\pi(a)} (j,\pi\cdot {s})$ in $\tl{A}$.

    Similarly, if $(i,r)\xto{\scriptnew a} (j,s)$ in $\tl{A}$ we have
    $(i,\ol{r})\xto{\scriptnew b} (j,\ol{s})$ in $A$ for some $b\fresh s$ and
    $\ol{r}, \ol{s}$ extending $r, (a\, b) s$. By equivariance of
    transitions in $A$ it follows that
    \[
      (i,\pi\cdot \ol{r})\xto{\scriptnew \pi(b)} (j,\pi\cdot \ol{s})
      \qquad\text{in $A$.}
    \]
    Then $\pi(b)\fresh \pi\cdot s$ and
    $\pi\cdot \ol{r},\pi\cdot \ol{s}$ extends
    $\pi\cdot {r}, (\pi(a)\, \pi(b)) \pi\cdot s$. Moreover
    \[
      \supp(\pi\cdot s)
      =
      \pi(\supp(s)) \seq \pi(\supp(r)\cup \{a\})
      =
      \supp(\pi\cdot r)\cup \{\pi(a)\},
    \]
    and so we conclude
    $(i,\pi\cdot {r})\xto{\scriptnew \pi(a)} (j,\pi\cdot {s})$ in $\tl{A}$.

  \item Transitions are $\alpha$-invariant: Suppose that
    $(i,r)\xto{\scriptnew a} (j,s)$ in $\tl{A}$ and
    $\braket{a} (j,s)=\braket{a'} (j,s')$ for some $a'\neq a$ and some
    $s'$. The former means that $(i,\ol{r})\xto{\scriptnew b} (j,\ol{s})$ in $A$
    for some $b\fresh s$ and some $\ol{r},\ol{s}$ extending
    $r,(a\,b)s$, and the latter that $s'=(a\, a')s$ and $a'\fresh
    s$. Then $\ol {s}$ extends $(a'\, b)s'$. Since $A$ is
    $\alpha$-invariant, we may assume that $b\neq a'$, i.e.\
    $b\fresh s'$. Moreover,
    \[
      \supp(s')
      \seq
      \supp(s)\setminus \{a\}\cup \{a'\}
      \seq
      (\supp(r)\cup \{a\})\setminus \{a\} \cup \{a'\}
      \seq
      \supp(r)\cup \{a'\}.
    \]
    Thus, $(i,r)\xto{\scriptnew a'} (j,s')$ in $\tl{A}$.
  
  \item Transitions are finitely branching: Let $(i,r)$ be a state of
    $\tl{A}$.

    For every outgoing transition $(i,r)\xto{a}(j,s)$ one has
    $a\in \supp(r)$ and $\supp(s)\seq \supp(r)$. Thus, there are only
    finitely many choices for~$a$ (since $\supp(r)$ is finite) and
    for~$s$ (since~$X_j$ and $\supp(r)$ are finite). Moreover, there
    are only finitely many choices for~$j$, since~$I$ is finite.

    Similarly, for every outgoing transition $(i,r)\xto{\scriptnew a} (j,s)$ one
    has $\supp(s)\seq \supp(r)\cup \{a\}$. The equivalence class
    $\braket{a} s$ is uniquely determined by~(1) the position
    $x_a\in X_j$ such that $s(x_a)=a$ (if any) and~(2) the domain
    restriction $s'$ of $s$ to $\dom(s)\setminus \{x_a\}$. There are
    only finitely many possible choices for~(1), and since
    $\supp(s')\seq \supp(r)$ there are only finitely many possible
    choices for~(2). Moreover, there are, again, only finitely many
    choices for~$j$, so the set
    $\{ \braket{a} s \compr (i,r)\xto{\scriptnew a}(j,s) \}$ is finite.\qedhere
  \end{enumerate}
\end{proof}

We first show the correctness of the name-dropping modification for the case of finite bar strings:

\begin{theorem}\label{thm:ncd_rnna}
  For every RNNA $A$, the literal language $L_0(\tl{A})$ is the closure of $L_0(A)$
  under $\alpha$-equivalence.
\end{theorem}
This follows from the two lemmas below:

\begin{lemma}\label{lem:ndc_alpha_closed}
  The literal language $L_0(\tl A)$ is closed under $\alpha$-equivalence.
\end{lemma}
\begin{proof}
  \begin{enumerate}
  \item\label{lem:ndc_alpha_closed:1} Let
    $\sigma_1\cdots \sigma_n\in \barA^*$ be a finite bar string and
    let $S_i=\free{\sigma_{i+1}\cdots \sigma_n}$ for
    $i=0,\ldots,n$. We claim that if
    \[
      (i_0,r_0)\xto{\sigma_1} (i_1, r_1)\xto{\sigma_2} \cdots \xto{\sigma_n} (i_n, r_n)
    \]
    is a run in $\tl{A}$ then also
    \[
      (i_0,r_0')\xto{\sigma_1} (i_1, r_1')\xto{\sigma_2} \cdots \xto{\sigma_n}
      (i_n, r_n')
    \]
    is a run in $\tl{A}$, where $r_i'$ is the restriction of $r_i$ to
    those $x\in \mathsf{dom}(r_i)$ such that $r_i(x)\in S_i$. (Note that
    $\supp(r_i')=S_i$ since $S_i\seq \supp(r_i)$ by \autoref{lem:rnna-props}.\ref{lem:rnna-props:3}.) The proof proceeds by
    induction on $n$.

    \smallskip\noindent The base case $n=0$ (i.e.~$w=\epsilon$) is
    trivial.
  
    \smallskip\noindent For the induction step, suppose that
    $n\geq 1$. If $\sigma_1=a$, then $S_0=S_1\cup \{a\}$ and so
    $\supp(r_1')\cup \{a\}=S_1\cup \{a\}\seq S_0 =\supp(r_0')$. Thus
    the transition $(i_0,r_0')\xto{a}(i_1,r_1')$ exists in
    $\tl{A}$. Similarly, if $\sigma_1=\newletter a$ we have
    $S_0=S_1\setminus \{a\}$, which implies $S_1\seq S_0\cup
    \{a\}$. Thus
    $\supp(r_1')=S_1 \seq S_0\cup \{a\} = \supp(r_0')\cup \{a\}$, so
    the transition $(i_0,r_0')\xto{\scriptnew a}(i_1,r_1')$ exists in
    $\tl{A}$. Now apply the induction hypothesis.

  \item Suppose that $w=_\alpha w'$. We show that every state $(i,r)$ of
    $\tl{A}$ accepting $w$ also accepts $w'$. We may assume w.l.o.g.\
    that $w=\newletter av$ and $w'=\newletter a'(a\, a')v$ for $a'\fresh v$. Moreover, by
    part \ref{lem:ndc_alpha_closed:1} we may assume that the
    first transition $(i,r)\xto{\scriptnew a} (j,s)$ of the accepting run
    for~$w$ is such that $\supp(s)$ is the set of free names of $v$;
    in particular, $a'\fresh s$. It follows that
    $\braket{a} s = \braket{a'} (a\, a')s$, so
    $(i,r)\xto{\scriptnew a'} (j,(a\, a')s)$ by $\alpha$-invariance. Since
    $(j,s)$ accepts $v$, the state $(j,(a\, a')s)$ accepts $(a\, a')v$
    by equivariance, and so we conclude that $(i,r)$ accepts
    the bar string $\newletter a'(a\, a')v$.\qedhere
  \end{enumerate}
\end{proof}
\begin{lemma}\label{lem:ndc_lang}
  The RNNAs $A$ and $\tl A$ accept the same bar language: $L_\alpha(A)=L_\alpha(\tl{A})$.
\end{lemma}
\begin{proof}
  ($\seq$) We show for every state $(i,r)$ of $A$ and every bar string
  $w\in \barA^*$: if $(i,r)$ accepts $w$ in $A$, then $(i,r)$ accepts
  some $w'\in \barA^*$ such that $w'=_\alpha w$ in $\tl{A}$. The proof
  is by induction on the length of $w$.

  \smallskip\noindent The base case $w=\epsilon$ is trivial. 

  \smallskip\noindent If $w=av$ for $a\in \names$ and $v\in \barA^*$,
  let $(i,r)\xto{a} (j,s)$ be the first transition of an accepting run
  for $w$ in $A$. Then the same transition exists in $\tl{A}$ and by
  induction, $(j,s)$ accepts some $v'\in \barA^*$ such that
  $v'=_\alpha v$ in $\tl{A}$. Thus $(i,r)$ accepts the word $w'=av'$
  in $\tl{A}$, and we have $w'=av'=_\alpha av = w$.

  \smallskip\noindent If $w=\newletter bv$ for some $b\in \names$ and
  $v\in \barA^*$, let $(i,r)\xto{\scriptnew b} (j,s)$ be the first
  transition of an accepting run for $w$ in $A$. Choose $a\in \names$
  such that $a\fresh s,v$. Then we have the transition
  $(i,r)\xto{\scriptnew a} (j,(a\, b)s)$ in $\tl{A}$. By induction,
  the state $(j,s)$ of $\tl{A}$ accepts some $v'=_\alpha v$. Thus
  $(j, (a\, b) s)$ accepts $(a\,b)v'$ by equivariance. It follows that
  $(i,r)$ accepts $w'=\newletter a(a\, b)v'$ in $\tl{A}$ and moreover
  $w=\newletter bv =_\alpha \newletter a(a\, b)v =_\alpha \newletter
  a(a\, b) v'=w'$.

  \medskip\noindent($\supseteq$) We show that for every state $(i,r)$
  of $\tl{A}$ accepting $w\in \barA^*$, there exists a state
  $(i,\ol{r})$ of $A$ such that
  \begin{enumerate}
  \item\label{cond:1} $\ol{r}(x)=r(x)$ for every $x\in \dom(r)$ with $r(x)$ free in
    $w$, and
  \item\label{cond:2} $(i,\ol{r})$ accepts some $w'\in \barA^*$ such that $w'=_\alpha w$.
  \end{enumerate}
  The proof is again by induction on the length of $w$.

  \smallskip\noindent The base case $w=\epsilon$ is trivial.

  \smallskip\noindent For the induction step, suppose that $(i,r)$
  accepts $w=\sigma v$ for $\sigma \in \barA$ and $v\in \barA^*$. Let
  $(i,r)\xto{\sigma} (j,s)$ be the first transition of an accepting
  run for $w$ in $\tl{A}$. By induction, there exists $(j,\ol{s})$ in
  $A$ such that
  \begin{enumerate}
  \item $\ol{s}(x)=s(x)$ for every $x\in \dom(s)$ with $s(x)$ free in $v$.
  \item $(j,\ol{s})$ accepts some $v'\in \barA^*$ such that $v'=_\alpha v$.
  \end{enumerate}
  We treat the cases of free and bound transitions separately:

  \smallskip\noindent\textbf{Case $\sigma=a$:} \\
  Since $(i,r)\xto{a} (j,s)$ in $\tl{A}$ we have
  $(i,\ol{r})\xto{a}(j,\ol{s}')$
  in $A$ for some
  $\ol{r}, \ol{s}'$ extending $r,s$. Then clearly the state
  $(i,\ol{r})$ satisfies~\ref{cond:1}. To prove condition~\ref{cond:2}, choose a
  permutation $\tau\in \Perm(\names)$ such that
  $\tau\cdot \ol{s}=\ol{s}'$; such a $\tau$ exists because
  $(j,\ol{s})$ and $(j,\ol{s}')$ are states of $A$, so
  $\ol{s},\ol{s}'\colon X_j\to \names$ are total injective maps,
  whence elements of $A^{\# X_j}$ which has just one orbit.
  We show that the word
  \[
    w':=a(\tau\cdot v')
  \]
  satisfies the desired condition. Indeed, since $(j,\ol{s})$ accepts
  $v'$, the state $(j,\ol{s}')=(j,\tau\cdot\ol{s})$ accepts
  $\tau\cdot v'$ by equivariance and thus $(i,\ol{r})$ accepts
  $w'$. Moreover, since the free variables of~$v'$ (equivalently,
  those of $v$) occur at the the same positions of $s$, $\ol{s}$ and
  $\ol{s}'$, the permutation~$\tau$ fixes them. Thus
  $\tau\cdot v' =_\alpha v'$ and therefore
  \[ w'=a(\tau\cdot v') =_\alpha av' =_\alpha av = w.\]

  \smallskip\noindent \textbf{Case $\sigma=\newletter a$:} \\
  Since $(i,r)\xto{\scriptnew a} (j,s)$ in $\tl{A}$ we have
  $(i,\ol{r})\xto{\scriptnew b}(j,\ol{s}')$ in $A$ for some
  $b\fresh s$ and some $\ol{r}, \ol{s}'$ extending $r,(a\,
  b)s$. Choose a name $c\fresh a,b,\ol{r},v'$ and put $\pi = (a\, c)$.
  Then we have the transition
  $(i,\pi\cdot \ol{r})\xto{\scriptnew \pi(b)} (j, \pi\cdot \ol{s}')$
  in ${A}$ by equivariance. The state $(i,\pi\cdot \ol{r})$ satisfies
  condition~\ref{cond:1} because $a$ is not free in $w=\newletter
  av$. To prove condition~\ref{cond:2}, choose a permutation
  $\tau\in \Perm(\names)$ such that
  $\tau\cdot \ol{s}=\pi \cdot \ol{s}'$ and $a\fresh \tau\cdot
  v'$. This choice is possible because $a\fresh \pi\cdot \ol{s}'$. We
  show that the word
  \[
    w':= \newletter \pi(b)(\tau\cdot v')
  \]
  satisfies the desired. Indeed, since $(j,\ol{s})$ accepts
  $v'$, the state $(j,\pi\cdot \ol{s}')=(j,\tau\cdot \ol{s})$
  accepts $\tau\cdot v'$ by equivariance and thus
  $(i,\pi\cdot \ol{r})$ accepts $w'$.

  We now prove that
  \begin{equation}\label{eq:tau}
    (a\, \pi(b))\tau\cdot v' =_\alpha v'.
  \end{equation}
  To see this we will show that the permutation
  $(a\, \pi(b))\tau$ fixes all free variables of~$v'$ (equivalently,
  those of~$v$). Indeed, if $a$ is free in $v$, then $a$ occurs in $s$; thus,
  at the same position $a$ occurs in $\ol{s}$ and $\pi(b)$ occurs at
  that position in $\pi\cdot \ol{s}'$. Since $\tau \cdot \ol s = \pi
  \cdot \ol{s}'$, we see that $\tau$ maps $a$ to $\pi(b)$, whence 
  $(a\, \pi(b))\tau$ maps $a$ to $a$. For every other free variable
  $a'\neq a$ of $v$, the name $a'$ occurs in $s$, $\ol{s}$,
  $\ol{s}', \pi\cdot \ol{s}'$ at the same position. Thus
  $\tau$ maps $a'$ to $a'$, which implies that $(a\, \pi(b))\cdot \tau$ maps
  $a'$ to $a'$ because $a'\neq a,\pi(b)$, using that
  $\pi(b)\in \{b,c\}$ and $b\fresh s$, $c\fresh v'$, so both $b,c$ are
  not free variables of $v$.

  Finally, we obtain 
  \[
    w'
    =
    \newletter \pi(b)(\tau\cdot v')
    =_\alpha
    \newletter a(a\, \pi(b))\tau\cdot v'
    \overset{\eqref{eq:tau}}{=_\alpha}
    \newletter av'
    =_\alpha
    \newletter av
    = w,
  \]
  where, in the second step, we use the definition of $\alpha$-equivalence
  and that $a\fresh \tau\cdot v'$. 
\end{proof}

We are ready to prove the correctness of the name-dropping modification for Büchi RNNAs (\autoref{thm:ncd_buechi_rnna}). As for the case of finite bar strings established above  (\autoref{thm:ncd_rnna}) we split the proof into two lemmas:
\begin{lemma}\label{lem:alpha_closure_buechi}
  The literal $\omega$-language of $\tl{A}$ is closed under $\alpha$-equivalence.
\end{lemma}
\begin{proof}
  \begin{enumerate}
  \item\label{lem:alpha_closure_buechi:1} Recall from \autoref{lem:ndc_alpha_closed} that $L_0(\tl{A})$
    is closed under $\alpha$-equivalence; in fact, the proof of that
    lemma shows that if
    $\sigma_1\sigma_1\cdots \sigma_n=_\alpha \sigma_1'\sigma_2'\cdots
    \sigma_n'$ and
    \[
      q_0\xto{\sigma_1} q_1\xto{\sigma_2} \cdots \xto{\sigma_n} q_n
    \]
    is a run for $\sigma_1\sigma_2\cdots \sigma_n$ in $\tl{A}$, then
    there exists a run
    \begin{equation}\label{eq:run}
      q_0\xto{\sigma_1'} q_1'\xto{\sigma_2'} \cdots \xto{\sigma_n'}
      q_n'
    \end{equation}
    in $\tl{A}$ such that $q_i'$ is final iff $q_i$ is final for
    $i=1,\ldots, n$.%
    
  \item Let $w=\sigma_1\sigma_1\sigma_2\cdots$ and
    $w'=\sigma_1'\sigma_2'\sigma_3'\cdots$ be two infinite bar strings
    such that $w=_\alpha w'$, and suppose that $\tl{A}$ accepts $w$
    via the accepting run
    \[
      q_0\xto{\sigma_1} q_1\xto{\sigma_2} q_2\xto{\sigma_3}\cdots.
    \]
    We need to construct an accepting run for $w'$. Consider the tree
    whose nodes are the partial runs \eqref{eq:run} where $n\in \Nat$
    and $q_i'$ final iff $q_i$ final for $i=1,\ldots, n$, and whose
    edge relation is given by extension of runs. Thus, the nodes of
    depth $n$ are precisely the runs \eqref{eq:run}. This tree is
    infinite since a run \eqref{eq:run} exists for each $n$ by
    part \ref{lem:alpha_closure_buechi:1}, but finitely branching since there are only finitely
    many such runs for each $n$, using that
    $\supp(q_i')\seq \supp(q_0)\cup \Names{\sigma_1'\sigma_2'\cdots
      \sigma_n'}$ for $i=1,\ldots,n$ by \autoref{lem:rnna-props} and that the orbit-finite set $\tl{Q}$
    contains only finitely many elements with a given support. Thus,
    the tree contains an infinite path by K\H{o}nig's lemma. That is,
    we obtain an infinite run
    \[
      q_0\xto{\sigma_1'} q_1'\xto{\sigma_2'} q_2'\xto{\sigma_3'}\cdots
    \]
    with $q_i'$ final iff $q_i$ final for each $i$. This is an
    accepting run for $w'$, as required.\qedhere
  \end{enumerate}
\end{proof}
\begin{lemma}\label{lem:alpha_lang_buechi}
  The Büchi RNNAs $A$ and $\tl A$ accept the same bar
  $\omega$-language.
\end{lemma}
\begin{proof}
  \begin{enumerate}
  \item Recall from \autoref{lem:ndc_lang} that
    $L_\alpha(A)=L_\alpha(\tl{A})$. In fact, the proof of that lemma
    shows that for every $\sigma_1\sigma_2\cdots \sigma_n\in \barA^*$
    and every run
    \[
      q_0\xto{\sigma_1} q_1\xto{\sigma_2} \cdots \xto{\sigma_n} q_n
    \]
    in $\tl{A}$ there exist $\sigma_1'\sigma_2'\cdots \sigma_n'\in \barA^*$ and a run
    \begin{equation}\label{eq:run2}
      q_0\xto{\sigma_1'} q_1'\xto{\sigma_2'} \cdots \xto{\sigma_n'}
      q_n'
    \end{equation}
    in ${A}$ such that
    $\sigma_1'\sigma_2'\cdots \sigma_n'=_\alpha \sigma_1\sigma_2\cdots
    \sigma_n$ and $q_i'$ final in $A$ iff $q_i$ final in $\tl{A}$ for
    $i=1,\ldots, n$; and vice versa with interchanged roles of $A$ and
    $\tl{A}$.
    
  \item To prove $L_{\alpha,\omega}(\tl{A})\seq L_{\alpha,\omega}(A)$,
    let $w=\sigma_1\sigma_2\sigma_3\cdots$ be an infinite bar string
    in $L_{0,\omega}(\tl{A})$ and 
    \[
      q_0\xto{\sigma_1} q_1\xto{\sigma_2} q_2 \xto{\sigma_3}\cdots
    \]
    be an accepting run for $w$ in $\tl{A}$; we need to find $w'\in
    L_{0,\omega}(A)$ such that $w'=_\alpha w$. Let
    $m=\degree(A)$ and choose a set $S\seq \names$ of
    $m+1$ distinct names containing
    $\supp(q_0)$. Consider the tree whose nodes are the partial runs
    \eqref{eq:run2} in $A$ where (1)~$n\in \Nat$,
    (2)~$q_i'$ final in $A$ iff
    $q_i$ final in $\tl{A}$ for $i=1,\ldots, n$, (3)~$\Names{\sigma_1'\sigma_2'\cdots \sigma_n'}\seq S$, and (4)~$\sigma_1'\sigma_2'\cdots \sigma_n'=_\alpha
    \sigma_1\sigma_2\cdots
    \sigma_n$. The edge relation is given by extension of runs. Thus,
    the nodes of depth
    $n$ are precisely the runs \eqref{eq:run2}. This tree is infinite
    since, by part 1 and the proof of \autoref{lem:restrict_num_names},
    for each $n\in
    \Nat$ there exists a run \eqref{eq:run2} satisfying the above
    conditions \mbox{(2)--(4)}. It is finitely branching because for
    every such run one has $\supp(q_i')\seq S$ and $\sigma_i'\in
    S$ for each
    $i$, so there are only finitely many nodes of depth
    $n$. Thus, the tree contains an infinite path by K\H{o}nig's
    lemma. That is, we obtain an infinite run
    \[
      q_0\xto{\sigma_1'} q_1'\xto{\sigma_2'} q_2'\xto{\sigma_3'}\cdots
    \]
    with $q_i'$ final iff $q_i$ final for each $i$, and
    $\sigma_1'\sigma_2'\cdots \sigma_n'=_\alpha \sigma_1\sigma_2\cdots
    \sigma_n$ for each $n$. Putting
    \[
      w':= \sigma_1'\sigma_2'\sigma_3'\cdots
    \]
    we obtain $w'=_\alpha w$ and $w\in L_{0,\omega}(A)$, as required.
    
  \item The proof of the reverse inclusion
    $L_{\alpha,\omega}(A)\seq L_{\alpha,\omega}(\tl{A})$ is
    symmetric.\qedhere
  \end{enumerate}
\end{proof}

This concludes the proof of \autoref{thm:ncd_buechi_rnna}.

\section{Details for \autoref{sec:decidability}}

\section*{Proof of \autoref{lem:word-inclusion}}
  ($\Leftarrow$)~Suppose that there exists
  $w'\sqsupseteq w$ such that $[w']_\alpha\in L$, and let $u\in
  D(w)$. Then there exists $v\in \barA^\omega$ such that $v=_\alpha w$
  and $\ub(v)=u$. By \autoref{lem:order-rel}, there exists
  $v'\sqsupseteq v$ such that $v'=_\alpha w'$. Thus
  $[v']_\alpha = [w']_\alpha\in L$, so $u=\ub(v)=\ub(v')\in D(L)$.

  \smallskip\noindent ($\Rightarrow$)~Suppose that $D(w)\seq D(L)$,
  and let $A$ be a Büchi RNNA accepting $L$ with the initial state
  $q_0$. By \autoref{thm:ncd_buechi_rnna} we may assume that
  $L_{0,\omega}(A)$ is closed under $\alpha$-equivalence. Choose a
  clean bar string $v\in \barA^\omega$ (\autoref{lem:clean}) such that
  $v=_\alpha w$ and such that for every $a\in \supp(q_0)$ the
  letter $\newletter a$ does not appear in $v$.  Then
  $\ub(v)\in D(w)\seq D(L)$, so there exists $v'\in \barA^\omega$ such
  that $\ub(v')=\ub(v)$ and $v'\in L_{0,\omega}(A)$.

  We now prove that $v\sqsubseteq v'$. To see this, let
  $v=\sigma_1\sigma_2\sigma_3\cdots$ and
  $v'=\sigma_1'\sigma_2'\sigma_3'\cdots$ and suppose that
  $\sigma_n=\newletter a$ for some $a\in \names$. Since
  $\ub(v')=\ub(v)$ we know that $\sigma_n'\in \{a,\newletter a\}$, and
  we need to prove that $\sigma_n'=\newletter a$. Suppose to the
  contrary that $\sigma_n'=a$. Then $a\fresh \supp(q_0)$ by the above
  choice of $v$. By \autoref{lem:rnna-props}.\ref{lem:rnna-props:3}
  the name $a$ cannot occur freely in $v'$. This implies that there
  exists $m<n$ such that $\sigma'_m=\newletter a$. Thus, $\ub(v')$ has
  the name $a$ at position~$m$ and~$n$. On the other hand, since
  $\sigma_n=\newletter a$ and $v$ is clean, $\sigma_m$ is not equal to
  $a$ or $\newletter a$, that is, $\ub(v')$ has distinct names at
  position $m$ and $n$. Thus $\ub(v')\neq \ub(v)$, a
  contradiction. 

  Since $v=_\alpha w$ and $v\sqsubseteq v'$, by
  \autoref{lem:order-rel} there exists $w'\sqsupseteq w$ such that
  $w'=_\alpha v'$. Then $[w']_\alpha=[v']_\alpha\in L$, as required.

\section{Details for \autoref{sec:relation}}

\section*{Proof of \autoref{prop:muller-buchi-equiv}}
  We show how to convert a Büchi RNNA into an equivalent Muller RNNA
  and vice versa. The two constructions are completely analogous to the
  classical ones.
  \begin{enumerate}
  \item If $A=(Q,R,q_0,F)$ is a Büchi RNNA accepting $L$, then the
    Muller RNNA $A'=(Q,R,q_0,\F)$ where
    $\F = \{ X\seq Q \compr X\cap F\neq \emptyset \}$ accepts $L$.
    
  \item Suppose that $A=(Q,R,q_0,\F)$ is a Büchi RNNA accepting $L$. Let $\F=\{F_1,\ldots, F_n\}$ where $F_i\seq \orb(Q)$, and denote by $\ol{F_i}\seq Q$ the union $\bigcup_{O\in F_i} O$. Then the following Büchi RNNA $A'=(Q',R',q_0,F)$ accepts $L$:
    \begin{itemize}
    \item $Q' = Q \cup \bigcup_{i=1}^n \{i\}\times \ol{F_i} \times \Pow(F_i)$.
    \item For each transition $q\xto{\sigma} q'$ in $A$ we have the following transitions in $A'$:
      \begin{enumerate}[(1)]
        \setlength{\itemindent}{0.35cm}
      \item $q\xto{\sigma}q'$;
      \item $q\xto{\sigma} (i,q',\emptyset)$ for all $i$ such that $q'\in \ol{F_i}$;
      \item $(i,q,R)\xto{\sigma} (i,q',R\cup \{q'\})$ for all $i$ with
        $q,q'\in \ol{F_i}$ and all $R\seq F_i$ with $R\neq F_i$;
        
      \item $(i,q,F_i)\xto{\sigma} (i,q',\emptyset)$ for all $i$ such that $q,q'\in \ol{F_i}$.
      \end{enumerate}
    \item $F=\{ (i,q, F_i) \compr i=1,\ldots, n, \, q\in \ol{F_i} \}$. 
    \end{itemize} 
  \end{enumerate}
  While simulating $A$ the automaton $A'$ guesses that from some point
  on, only states from $\ol{F_i}$ for some $i$ will be visited, and
  uses the third component of states to make sure that each orbit from
  $F_i$ is visited infinitely often.

\end{document}